%% file: main.tex
\documentclass[11pt,twoside]{article}
\usepackage[letterpaper,margin=1in]{geometry}
\usepackage{lineno}
\usepackage{comment}
\usepackage{amsmath}
\usepackage{amsthm}
\usepackage[charter]{mathdesign}
\usepackage{eulervm} 
\usepackage{mathtools}
\usepackage{tikz}
\usepackage[boxruled,vlined,linesnumbered]{algorithm2e}
\SetKwRepeat{Do}{do}{while}
\usepackage{enumitem}
\usepackage[linktocpage=true,pagebackref=true,colorlinks,linkcolor=magenta,citecolor=blue,bookmarks,bookmarksopen,bookmarksnumbered]{hyperref}
\usepackage[nameinlink]{cleveref} 
\usepackage{bm}
\usetikzlibrary{arrows.meta, calc, positioning}
\usepackage{booktabs} 
\usepackage{array}
\usepackage{titlesec}
\titlespacing*{\section}{0pt}{1.5ex plus 0.5ex minus 0.2ex}{1ex plus 0.2ex}
\titlespacing*{\subsection}{0pt}{1.2ex plus 0.5ex minus 0.2ex}{0.8ex plus 0.2ex}

\newtheorem{theorem}{Theorem}[section]
\newtheorem{lemma}[theorem]{Lemma}

\newtheorem{observation}[theorem]{Observation}

\newtheorem{definition}[theorem]{Definition}

\newtheorem{assumption}[theorem]{Assumption}

\renewcommand{\mod}[1]{{|#1|}}
\newcommand{\al}{\alpha}
\newcommand{\be}{{0.875}}
\newcommand{\ga}{\gamma}
\newcommand{\ep}{\varepsilon}

\newcommand{\uf}{\textsc{UpdateFrom}}
\newcommand{\fc}{\textsc{FindCenters}}

\newcommand{\ofc}{\textsc{ModifiedFindCenters}}
\newcommand{\fbc}{\textsc{FindBalls\&Clusters}}
\newcommand{\smp}{\textsc{Sample}}
\newcommand{\floor}[1]{\left\lfloor#1\right\rfloor}
\newcommand{\sttwo}{\frac{\mod{st}}{2}}
\newcommand{\Ot}{\tilde{{O}}}

\newcommand{\tl}{\tilde{{O}}}
\newcommand{\ballal}{ball_{\al}}

\newcommand{\clal}{cluster_{\al}}
\newcommand{\two}[1]{2^{2^{#1}}}
\newcommand{\prob}[1]{\frac{1}{#1}}
\newcommand{\lb}{\left(}
\newcommand{\rb}{\right)}

\newcommand{\lgn}{{\log\log n}}
\newcommand{\lgnone}{{\log\log n-1}}
\newcommand{\tri}{\textsc{Triangulate}}
\newcommand{\clu}{{cluster}}
\newcommand{\ba}{{ball}}
\newcommand{\clo}{\textsc{EstimatedClose}}
\newcommand{\tclo}{\textsc{TrueClose}}
\newcommand{\bc}{\textsc{BuildClose}}
\newcommand{\bp}{\textsc{EnsureBallProximity}}
\newcommand{\ap}{\textsc{Apsp}}
\newcommand{\fmm}{\textsc{Fmm}}
\newcommand{\ex}{\mathbb{E}}

\author{
  \begin{minipage}[t]{0.45\textwidth}
    \centering
    Manoj Gupta\thanks{Research supported by the Anusandhan National Research Foundation (ANRF) through MATRICS grant ANRF/ARGM/2025/001853/MTR.} \\
    IIT Gandhinagar \\
    India \\
    {\small\texttt{gmanoj@iitgn.ac.in}}
  \end{minipage}
  \hfill
  \begin{minipage}[t]{0.45\textwidth}
    \centering
    Mrigankashekhar Shandilya\footnote{Now at Google, India} \\
    IIT Gandhinagar \\
    India \\
    {\scriptsize\texttt{mrigankashekhar.shandilya@alumni.iitgn.ac.in}}
  \end{minipage}
}
\title{$\tilde{O}$ptimal Algorithm for 2-Approximate All Pair Shortest Paths --- almost}
\date{}

\begin{document}

\maketitle
\input{abstract}
\input{intro}

\input{overview}
\input{prelims}
\input{assumption}

\input{algo}

\bibliographystyle{alpha}

\bibliography{references.bib}
\input{appendix}

\section*{AI Disclosure}
AI Disclosure: We used an AI language model (e.g., Gemini/ChatGPT/Claude) to assist with proofreading, improving clarity, and checking the logical correctness of claims/equations.  The authors verified the correctness and originality of all content including references.

\end{document}

%% file: abstract.tex
\begin{abstract}
Given an undirected, unweighted graph $G$, we aim to compute a 2-approximation of all-pairs shortest paths ($\ap$). This problem admits a natural lower bound of $\Omega(n^2)$ since the output size is $\Theta(n^2)$. A central goal in this area is to achieve a running time of $O(n^2)$.

Dor, Halperin, and Zwick (FOCS 1996, SICOMP 2001) designed an algorithm with a running time\footnote{$\tl$ hides poly$\log n$ factors.} of $\tilde{O}(n^2)$ that guarantees a 2-approximation only for pairs at a distance of at least $O(\log n)$. Recently, Gupta (FOCS 2025) improved this bound, handling all pairs at a distance of at least $O(\log \log n)$. 

We nearly resolve this problem. We design a randomized algorithm that runs in $\tilde{O}(n^2)$ time and, with high probability\footnote{With probability $\ge 1-1/n^{k}$ for a constant $k$.}, guarantees a 2-approximation for all pairs at distance at least $c$, where $c \ge 0$ is a constant.  Unlike the above two results, which were purely combinatorial, our algorithm combines combinatorial techniques with fast matrix multiplication ($\fmm$).
\end{abstract}

%% file: intro.tex
\section{Introduction}

Let $G = (V, E)$ be an unweighted, undirected graph with $n$ vertices and $m$ edges. Let $st$ denote a shortest path between vertices $s$ and $t$ in $G$, and let $|st|$ denote its length. Computing all pairs shortest paths in a graph $G$ is a fundamental problem in graph algorithms. For unweighted graphs, a simple combinatorial algorithm computes all pairs shortest paths in $O(mn)$ time. Seidel \cite{Seidel95} designed an algorithm that uses fast matrix multiplication $(\fmm)$ and solved the problem in $\tl(n^{\omega})$ time, where $n^{\omega}$ denotes the time required to multiply two $n \times n$ matrices. Even with the advancements in fast matrix multiplication in  \cite{AlmanW24,Gall14,Stothers10,Williams12,DuanHR23,WilliamsXZR24}, this running time remains superquadratic. To achieve near quadratic running time, we must allow approximate solutions.

Let $est(s,t)$ denote the estimate of the shortest path between $s$ and $t$ returned by an algorithm. We say that the algorithm is $(\al,\beta)$ approximate, for $\al \ge 1$ and $\beta \ge 0$, if $$|st| \le est(s,t) \le \al |st| + \beta, \hspace{2cm} \text{for all $s,t$ pairs in $G$}$$

If $\beta = 0$, we obtain a purely multiplicative approximation and call the algorithm an $\al$ approximate all pairs shortest paths algorithm, or $\al$-$\ap$ algorithm in short. If $\al = 1$, we obtain a purely additive approximation and call the algorithm a $+\beta$-$\ap$ algorithm. In this paper, we focus on purely multiplicative, or $\al$-$\ap$.

Dor, Halperin and Zwick \cite{DorHZ00} showed that obtaining $(2-\ep)$-$\ap$ is as hard as Boolean Matrix Multiplication. Consequently, obtaining a nearly quadratic algorithm for $\al$-$\ap$ with $\al < 2$ would imply a major breakthrough in matrix multiplication (i.e., $\omega = 2$). But what about the case when $\al=2$ or 2-$\ap$? The only natural lower bound is the output size, which is $\Theta(n^2)$. Thus, any 2-$\ap$ algorithm requires $\Omega(n^2)$ time. Since no stronger lower bound is known, an important problem is to design a $2$-$\ap$ algorithm that runs in $O(n^2)$ time.

 We will see later that purely additive $\ap$ algorithms are used to obtain 2-$\ap$. We now describe results for purely additive $\ap$. Aingworth, Chekuri, Indyk and Motwani \cite{AingworthCIM99} designed a $\tl(n^{2.5})$ time algorithm for +2-$\ap$. Later, Dor, Halperin and Zwick \cite{DorHZ00} reduced the running time to $\tl(\min\{n^{3/2}m^{1/2},n^{7/3}\})$. Using $\fmm$,  Deng, Kirkpatrick, Rong, Williams and Zhong \cite{DengYRWZ22} improved the running time to $O(n^{2.259})$ (using an improved result of \cite{Durr23}). Very recently, Jin, Kirkpatrick, Stawarz and Williams \cite{JinKSW26}, again using $\fmm$, improved this running time to $O(n^{2.2255})$.

 We now describe a simple relation that yields a $2$-$\ap$ from a purely additive approximation.
 
 \begin{observation}
 \label{obs:additive}
 	A $+k$-$\ap$ implies a $2$-$\ap$ for all pairs of vertices at distance at least $k$.
 \end{observation}
 
 Since we can find a path of length one in $O(m)$ time, the above observation translates the +2-$\ap$ from the previous paragraph to 2-$\ap$. However, all these algorithms require $\omega(n^2)$ time.  Fortunately, there are $+k$-$\ap$ algorithms with better running times. Dor, Halperin and Zwick \cite{DorHZ00} designed a $+k$-$\ap$ algorithm with running time $\tl(\min\{n^{2-\frac{2}{k+2}}m^{\frac{2}{k+2}},n^{2+\frac{2}{3k-2}}\})$. Saha and Ye \cite{SahaY24} recently used $\fmm$ to improve the running time. However, the worst case running time in these results remain $\tl(n^{2+O(\prob{k})})$. If $k= \log n$, then the running time becomes $\tl(n^2)$ and the algorithm works for all pairs of vertices at a distance at least $O(\log n)$. Until recently, this was the best result for $2$-$\ap$ with the running time $\tl(n^2)$.  
  
 Gupta \cite{Gup25} designed an algorithm for $2$-$\ap$ that runs in $\tl(n^{2+1/k})$ time and works for vertices at distance at least $O(\log k)$.  If $k= \log n$, then the running time becomes $\tl(n^2)$, and the algorithm works for all pairs at a distance at least $O( \lgn)$. An important open problem in this field is to design a $2$-$\ap$ algorithm that works for all pairs and runs in $O(n^2)$ time or nearly optimal $\tl(n^2)$ time. We nearly resolve this problem by showing the following theorem.
 
 \begin{theorem}
 \label{thm:maintheorem}
 	There exists a randomized algorithm that, with high probability, runs in $\tl(n^2)$ time  and computes $2$-$\ap$ for all pairs of vertices at distance at least $c$, where $c \ge 0$ is a constant.
 \end{theorem}  
 
We do not attempt to optimize the constant, though we achieve $c \le 906$. Our approach builds on the combinatorial techniques developed by Gupta \cite{Gup25}. Our algorithm also builds some combinatorial techniques but a crucial step of our algorithm uses $\fmm$. The primary contribution of our work is showing how to combine these combinatorial techniques with $\fmm$.
 
\begin{table}[hpt!]
\centering
\renewcommand{\arraystretch}{1}
\begin{tabular}{c|>{\centering\arraybackslash}p{5cm}|c}

Running Time & Works for pairs that are at least at a distance & Reference \\
\midrule

$\tl(n^{3/2} m^{1/2})$ 
& all pairs
& \cite{CohenZ01}\\
\midrule

$\tl(m\sqrt{n} + n^2)$ 
& all pairs
& \cite{BasKav06}\\
\midrule

$\tl(n^{2.25})$ 
& all pairs
& \cite{Roditty23}\\
\midrule

$\tl(n^{2.031})$ 
& all pairs
& \cite{DoryFKNWV24, SahaY24}\\
\midrule

$\tl(n^{2})$ 
& $O(\log n)$
& \cite{DorHZ00}\\
\midrule

$\tl(n^{2})$ 
& $O(\log \log n)$
& \cite{Gup25}\\
\midrule
$\tl(n^2)$ 
& a constant $c$
& \Cref{thm:maintheorem}\\
\midrule

\end{tabular}
\caption{A snapshot of relevant results for 2-$\ap$.}
\end{table}

 \subsection{Other related results}
 There is another class of algorithms that do not use \Cref{obs:additive} and compute 2-$\ap$ directly. These algorithms work for all pairs of vertices. Cohen and Zwick \cite{CohenZ01} designed an $\tl(n^{3/2}m^{1/2})$ time algorithm for 2-$\ap$. Baswana and Kavitha \cite{BasKav06} achieved a running time of $\tl(m\sqrt n+n^2)$. Roditty \cite{Roditty23} designed a combinatorial algorithm for 2-$\ap$ in $\tl(n^{2.25})$ time. All of the above algorithms are combinatorial in nature. Using $\fmm$, Dory, Forster, Kirkpatrick, Nazari, Williams, and Vos \cite{DoryFKNWV24} and Saha and Ye \cite{SahaY24} improved the running time to $\tl(n^{2.031})$. There are many combinatorial algorithms \cite{BerShiv07, Som16, Knud17, PatRod10} that compute a $(2,1)$-$\ap$ in $\tl(n^2)$ time.
  
 \subsubsection*{Weighted graphs}
The algorithms of \cite{CohenZ01,BasKav06} work for weighted graphs. The authors \cite{CohenZ01,BasKav06} also designed a $3$-$\ap$ algorithm that runs in $\tl(n^2)$ time. Kavitha \cite{Kavitha12} designed a $(2+\ep)$-$\ap$ algorithm for weighted graphs running in $\tl(n^{2.25})$ time. Baswana and Kavitha \cite{BasKav06} designed an algorithm that computes a $(2, W_{s,t})$-$\ap$ in $\tl(n^2)$ time, where $W_{s,t}$ is the maximum edge weight on the shortest path between $s$ and $t$.

%% file: overview.tex
\section{Overview}

Using \Cref{obs:additive}, a $+O(1)$-$\ap$   implies a $2$-approximation for all pairs at a constant distance from each other.
In the overview, we describe the intuition behind an algorithm that provides a constant additive approximation for the shortest path between $s$ and $t$.  However, our algorithm relies on certain structural assumptions that do not hold in general. We first introduce the notation used in this overview and throughout the paper.

\subsection{Basic notations}

\begin{definition}[Estimate between a pair of vertices]
\label{def:est}
For every pair $(s,t)$ we maintain $est(s,t)$, which is our current best estimate of the length of the $st$ path. Several known algorithms compute $(2,1)$-approximate all-pairs shortest paths in $\tl(n^2)$ time \cite{BerShiv07, Som16, Knud17, PatRod10}. We therefore assume that $|st| \le est(s,t) \le 2|st| + 1$ is maintained throughout the algorithm. Whenever we update $est(s,t)$, we also update $est(t,s)$.
\end{definition}

Let $A$ be a set of vertices obtained by sampling each vertex of $V$ independently with probability $\prob{n^{0.05}}$. Hence, with high probability, $|A| = \tl(n^{0.95})$. The exponent $0.05$ does not play any special role in the analysis. Any sufficiently small constant exponent would suffice for our purposes. We now introduce some additional notations.

\begin{definition}[Pivots, Balls and Clusters]
For each $s \in V$, let $pivot_A(s)$ denote the vertex in $A$ closest to $s$ in $G$, breaking ties arbitrarily. Let $ball_A(s)$ be the set of vertices strictly closer to $s$ than $pivot_A(s)$. Formally,
$$ ball_A(s) = \{ v \mid |sv| < |s\,pivot_A(s)| \} $$
For each $v \in V$, let $\clu_A(v)$ be the set of vertices whose ball contains $v$. Formally,
$$ \clu_A(v) = \{ s \mid v \in ball_A(s) \} $$
\end{definition}

Balls and clusters are well known concepts in the literature on approximate shortest paths. One can show that $|ball_A(s)| = \tl(n^{0.05})$ with high probability. However, $|\clu_A(v)|$ may be large. Thorup and Zwick \cite{ThorupZ01} presented a sampling method that  bounds the size of clusters as well. For the overview, we assume the following.

\begin{assumption}
	\label{ass:overview:1}
	For each $s \in V$, $|ball_A(s)| = \tl(n^{0.05})$ and $|\clu_A(s)| = \tl(n^{0.05})$.
\end{assumption}

Observe that the ball contains all vertices at distance $ \le |s~pivot_A(s)| - 1$ from $s$. No vertex strictly inside the ball, that is, at distance $ \le |s~pivot_A(s)| - 2$ from $s$, has degree greater than $\tl(n^{0.05})$, since otherwise the ball size exceeds $\tl(n^{0.05})$. Thus, we obtain the following lemma.
\begin{lemma}
\label{lem:overview:insdieball}
	For each $s \in V$, every vertex strictly inside $ball_A(s)$, that is, at distance $ \le |s~pivot_A(s)| - 2$ from $s$, has degree $ \le \tl(n^{0.05})$.
\end{lemma}

\subsection{A new idea}
We now present a new approach to this problem. To connect the local balls of $s$ and $t$ to the global shortest path, we examine where the $st$ path crosses the boundaries of their respective balls.

\begin{definition}
\label{def:overview:xy}
Assuming $t \notin ball_A(s)$ and $s \notin ball_A(t)$, let $p_s$ denote the vertex on the boundary of $ball_A(s)$ that lies on the $st$ path. Similarly, let $p_t$ denote the vertex on the boundary of $ball_A(t)$ that lies on the $st$ path. See \Cref{fig:xy}. Thus $|sp_s| = |s\,pivot_A(s)| - 1$ and $|tp_t| = |t\,pivot_A(t)| - 1$.
\end{definition}

\begin{figure}[hpt!]
\centering
\begin{tikzpicture}[scale=1.1]

\node (s) at (0,0) {$s$};
\node (x) at (1.4,0.2) {$p_s$};
\node (y) at (2.6,0.2) {$p_t$};
\node (t) at (4,0) {$t$};

\draw (s) -- (t);

\draw[dashed] (s) circle (1.2);
\draw[dashed] (t) circle (1.2);
\fill (1.2,0) circle (2pt);
\fill (2.8,0) circle (2pt);
\node at (0.1,1.4) {$ball_{A}(s)$};
\node at (4,1.4) {$ball_{A}(t)$};

\end{tikzpicture}
\caption{Vertices $p_s$ and $p_t$ lie on the boundary of $ball_{A}(s)$ and $ball_{A}(t)$ respectively and lie on the $st$ path.}
\label{fig:xy}
\end{figure}
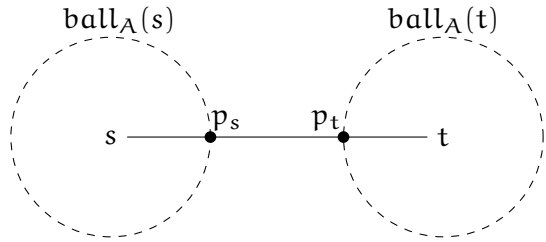

Using \Cref{lem:overview:insdieball}, each vertex on the $sp_s$ path, except possibly $p_s$, has degree $\le \tl(n^{0.05})$. We make the same argument for the vertices on the $p_t t$ path. Thus, the prefix $sp_s$ and the suffix $p_t t$ of the $st$ path have low degree, and we will use this crucially to design faster algorithms. But what about the middle segment, $p_sp_t$? Unfortunately, $p_sp_t$ may contain vertices with high degree. But what if we could show that $p_sp_t$ is short? This is exactly our main assumption in this section.

\begin{assumption}
	\label{ass:overview:2}
	$|p_sp_t| \le c$, where $c$ is a constant.
\end{assumption}

We will prove this assumption in later sections. But how does it help us find the $st$ path? By definition, $s \in \clu_A(p_s)$ and $t \in \clu_A(p_t)$. Unfortunately, given $s$ and $t$, we do not know the identities of $p_s$ and $p_t$. But what if we just guess them? Consider the following simple algorithm that iterates over all pairs $x$ and $y$, treating them as candidates for $p_s$ and $p_t$, and checks the paths between their clusters.

\begin{algorithm}[H]
	\caption{Trivial Pairwise Update}
	\label{alg:overview:trivial}
	\ForEach{$x,y$ pair in $G$}
	{
		\ForEach{$(s,t)$ in $\clu_A(x) \times \clu_A(y)$}
		{
			$
			est(s,t) = \min
			\begin{cases}
			est(s,t),\\
			est(s,x) + est(x,y) + est(y,t)
			\end{cases}
			$
		}
	}
\end{algorithm}

Since the algorithm checks all candidate pairs $x$ and $y$, it will eventually process the specific iteration where $x = p_s$ and $y = p_t$. Consider this specific pair $(p_s,p_t)$.  Because our algorithm takes the minimum over all pairs, checking the specific pair $(p_s, p_t)$ gives us an upper bound on $est(s,t)$:

\begin{align}
	est(s,t) &\le est(s,p_s) + est(p_s,p_t) + est(p_t, t) \notag\\
	\shortintertext{ Later in our paper, we will show that we can find the exact distance between $s$ and all the vertices in $ball_A(s)$ quickly. Thus, $est(s,p_s) = |sp_s|$. Similarly, $est(t,p_t) = |p_t t|$. Thus, we get:}
			 &= |sp_s| + est(p_s,p_t) + |p_t t|\notag\\
	\shortintertext{ By \Cref{def:est}, $est(p_s,p_t) \le 2|p_sp_t|+1$. Thus, we get:}
			 &\le |sp_s| + (2|p_sp_t|+1) + |p_tt|\notag\\
			 &= (|sp_s| + |p_sp_t| + |p_tt|) + |p_sp_t| + 1\notag\\
	\shortintertext{Since $p_s$ and $p_t$ lie exactly on the shortest $st$ path, $|sp_s| + |p_sp_t| + |p_tt| = |st|$. Using \Cref{ass:overview:2}, $|p_sp_t| \le c$. Thus, we get:}
			 &\le |st| + c+1 \label{eq:overview:approx}
\end{align}

Thus, we get a constant additive approximation of the $st$ path. Let us now look at the running time of our algorithm. Using \Cref{ass:overview:1}, $|\clu_A(x)|, |\clu_A(y)| = \tl(n^{0.05})$. Thus, the running time of our algorithm is bounded by:
$$
\displaystyle
\sum_{x,y}\sum_{\substack{s \in \clu_A(x), \\ t \in \clu_A(y)}} O(1) = \sum_{x,y} 
\tl(n^{0.05} \times n^{0.05}) = \tl(n^{2.1})
$$

\Cref{alg:overview:trivial} processes all pairs $x$ and $y$, resulting in an undesirable $\tl(n^{2.1})$ running time. But do we really need to check all pairs? By \Cref{ass:overview:2}, we know that our target vertices $p_s$ and $p_t$ are at most a distance $c$ apart. Unfortunately, our current algorithm blindly pairs vertices that are arbitrarily far from each other. But what if, for a given $x$, we only check candidates $y$ that are close to it? We define the following set to restrict our search space.

\begin{definition}
For each $x \in V$, we compute a set $\clo(x)$ defined as follows:
\[
\clo(x) = \{y \mid est(x,y) \le 2c+1\}
\]    
\end{definition}

Using \Cref{ass:overview:2}, $|p_sp_t| \le c$. Since we have pre-calculated a $(2,1)$-$\ap$, $est(p_s,p_t) \le 2c+1$. Thus, $\clo(p_s)$ contains $p_t$. 

We now modify  \Cref{alg:overview:trivial}: instead of pairing $x$ with every vertex in the graph, we pair it only with vertices in $\clo(x)$. This modification improves the running time when $\clo(x)$ is small. Unfortunately, if $|\clo(x)| = O(n)$, then the algorithm still performs as poorly as \Cref{alg:overview:trivial}. But what if we use the size of $\clo(x)$ to our advantage? We now present our final insight. For each candidate vertex $x$, we divide the algorithm into two cases based on size of $\clo(x)$:

\begin{enumerate}
	\item \textbf{Sparse Case: $|\clo(x)| \le n^{0.875}$}

	In this case, we use a modified version of \Cref{alg:overview:trivial}. We only pair $x$ with vertices $y \in \clo(x)$. Since the size of $\clo(x)$ is bounded, the number of pairs processed for a given $x$ reduces to $O(n^{0.875})$. Across all sparse vertices, this takes time:
	
	\begin{align*}
		\sum_{x} &\sum_{\substack{y\in \clo(x)}}\sum_{\substack{s \in \clu_A(x), \\ t \in \clu_A(y)}} O(1)\\ 
		&= \sum_{x} \sum_{\substack{y\in \clo(x)}} \tl(n^{0.05} \times n^{0.05}) \\
	  	&= \sum_{x} \tl(n^{0.875} \times n^{0.1})	\\
	  	& = \tl(n^{1.975})
	\end{align*}

	Thus, the running time of our algorithm is  less than $\tl(n^2)$. We now move to the second case.
	\item \textbf{Dense Case: $|\clo(x)| > n^{0.875}$}
		
	In this case, we exploit the fact that $\clo(x)$ is large. Let $L$ be a set of vertices sampled from $V$ with probability $\tl\lb\prob{n^{0.875}}\rb$. With high probability, $|L| = \tl(n^{0.125})$, meaning $L$ is a very small set. Moreover, with high probability, $L$ contains at least one vertex $z \in \clo(x)$. We therefore estimate the distance  using the vertices of $L$ as intermediaries:

	$$est(s,t) = \min_{z \in L} \{ est(s,z) + est(z,t) \}$$
	
	Later, we show that $est(s,t)$ is a constant additive approximation of $|st|$ and that we can compute these estimates concurrently for all dense vertices using Fast Matrix Multiplication ($\fmm$). Because $|L| = \tl(n^{0.125})$, this computation involves multiplying small rectangular matrices of size $n \times n^{0.125}$ and $n^{0.125} \times n$. Following \cite{Coppersmith82,Williams14ACC}, we can multiply these matrices in $\tl(n^2)$ time. Thus, the running time of our algorithm remains $\tl(n^2)$ despite the use of $\fmm$.
	
\end{enumerate}

Because the algorithm correctly handles every candidate $x$, it successfully approximates the $st$ path when it processes $x = p_s$. If $\clo(p_s)$ is sparse, the modified algorithm  checks the pair $(p_s, p_t)$ and computes an additive approximation of the $st$ path. If $\clo(p_s)$ is dense, one of the vertices in the random set $L$ lies near the $st$ path, and the $\fmm$ step guarantees we find a good approximation. 

To summarise, the key takeaway from this overview is the following. The $st$ path naturally decomposes into three segments: the low-degree prefix $sp_s$, the short middle segment $p_sp_t$, and the low-degree suffix $p_tt$. The prefix and suffix are handled cheaply because their low-degree structure allows efficient exact distance computation inside balls. The middle segment is handled by the sparse/dense case split: if few vertices are close to $p_s$, we enumerate them directly; if many are, a small random sample is guaranteed to hit one, and $\fmm$ does the rest. The entire approach hinges on \Cref{ass:overview:2}, which ensures the middle segment is short. Justifying this assumption is the main technical contribution of the subsequent sections, and we use the recent work of Gupta \cite{Gup25} to establish it.

%% file: prelims.tex
\section{Preliminaries}
Let $G$ be an undirected unweighted graph with $n$ vertices and $m$ edges. Let $V$ and $E$ denote the vertex and edge sets of $G$. For any pair $s,t$, let $st$ denote a shortest path between $s$ and $t$. If multiple shortest paths exist, pick one arbitrarily. Let $\mod{st}$ denote the length of $st$. Let $\deg(s)$ denote the degree of a vertex $s \in V$. We denote an edge between two adjacent vertices $s$ and $t$ by $(s,t)$. For brevity, we will assume that $\lgn$ is an integer.

Although our problem is defined on an unweighted graph, our algorithm constructs several auxiliary graphs that may contain weighted edges. We denote a weighted edge from $s$ to $t$ in such graphs by $[s,t]$ and denote its weight by $\mod{[s,t]}$. The degree of an edge $e=(s,t)$ is $\deg(e) := \min\{\deg(s),\deg(t)\}$. For two paths $P$ and $Q$ where $P$ ends at the vertex where $Q$ starts, we denote by $P \odot Q$ the concatenation of $P$ and $Q$.

Baswana and Kavitha \cite{BasKav06} proposed an $\Ot(m\sqrt{n} +n^2)$ time algorithm for $2$-$\ap$. We can use this result to find 2-approximate shortest paths for pairs connected entirely by low-degree vertices (degree at most $\Ot(\sqrt{n})$). To achieve this, we construct a subgraph $H$ of $G$ containing only vertices with degree at most $\Ot(\sqrt{n})$. Note that all low-degree paths from $G$ are preserved in $H$. Since the maximum degree in $H$ is $\Ot(\sqrt{n})$, the number of edges in $H$ is at most $\Ot(n\sqrt{n})$. Applying the algorithm from \cite{BasKav06} to $H$ takes $\Ot(n^2)$ overall time. We therefore assume the following for the rest of the paper.

\begin{assumption}
\label{ass:baswana}
	The shortest $st$ path contains at least one vertex with degree $\ge \Ot(\sqrt{n})$.
\end{assumption}

In our paper we use the algorithm of Gupta \cite{Gup25}. We now introduce some notations from it.

\begin{definition}[Nested Vertex Sampling]
\label{def:nested}
Let $A_0 = V$. Let $A_1$ be the set of vertices sampled from $A_0$ with probability $\prob{\two{1}}$. For each $2 \le i \le \lgn-1$, $A_i$ is sampled from $A_{i-1}$ with probability $\prob{\two{i-1}}$. 
\end{definition}

One can show that $\mod{A_i} = \tl(\frac{n}{\two{i}})$ with high probability. We now move on to the other notations used in \cite{Gup25}.

\begin{definition}[Pivots, Balls and Clusters]
\label{def:pbc}
For each $s \in V$, let $pivot_i(s)$ denote the vertex in $A_i$ closest to $s$ in $G$, breaking ties arbitrarily. Let $\ba_i(s)$ be the set of vertices strictly closer to $s$ than $pivot_i(s)$. Formally,
$$ \ba_i(s) = \{ v \mid \mod{sv} < \mod{s\,pivot_i(s)} \} $$
For each $v \in V$, let $\clu_i(v)$ be the set of vertices whose ball contains $v$. Formally,
$$ \clu_i(v) = \{ s \mid v \in \ba_i(s) \} $$
\end{definition}

 We now state the following lemma in \cite{Gup25} which bounds the size of balls. 

\begin{lemma}[\cite{Gup25}]
\label{lem:ball}
For each $0 \le i \le \lgnone$ and each $s \in V$, $\mod{\ba_i(s)} = \tl(\two{i})$ with a high probability. Moreover, we can compute the following in $\tl(n^2)$ time with a high probability:
\begin{enumerate}
	\item $pivot_i(s)$ and its distance from $s$ for all $s \in V$. So, $est(s, pivot_i(s)) = |s~pivot_i(s)|$.
	\item $\ba_i(s)$ for all $i$ and $s \in V$. Moreover, $est(s,v) = |sv|$ for each $v \in ball_i(s)$.
\end{enumerate}  
\end{lemma}

The above lemma bounds the number of vertices in $\ba_i(s)$. Furthermore, for each $i$ and $s$, we can efficiently compute $pivot_i(s)$ and $\ba_i(s)$. Additionally, we can compute the length of the shortest path from $s$ to each vertex in $\ba_i(s)$.

The above lemma also allows us to compute clusters, but it does not bound their size. The work \cite{Gup25} does not use clusters. In contrast, our algorithm uses clusters crucially and requires a bound on the cluster size. However, nested vertex sampling does not guarantee bounded cluster size. To address this issue, we adapt a result of Thorup and Zwick \cite{ThorupZ01}.

\begin{lemma}
	\label{lem:thorup}
	Fix an integer $\al$ where $1 \le \al \le \lgnone$. There is a randomised algorithm that finds a set $A_{\al}$ such that with high probability 
	\begin{enumerate} 
	\item $\mod{A_\al} = \tl(\frac{n}{\two{\al}})$
	\item $\mod{\ballal(s)} = \tl(\two{\al})$ and $\mod{\clal(s)} = \tl(\two{\al})$ for each $s \in V$
	\item The running time of the algorithm is $\tl(n^2)$ 
	\end{enumerate}
\end{lemma}

For completeness, we prove this lemma in \Cref{proof:thorup}. The above algorithm finds a set $A_\al$ such that every vertex has bounded ball and cluster size. In \Cref{def:nested} we start the sampling process from $A_1$. We now modify this step. Instead of starting from $A_1$, we start from $A_\al$ where $1 \le \al \le \lgnone$ and obtain $A_\al$ using \Cref{lem:thorup}. After computing $A_\al$, we apply nested vertex sampling to construct $\{A_{\al+1},A_{\al+2},\dots,A_{\lgnone}\}$. We describe the modified nested vertex sampling process below.

\begin{definition}[Modified Nested Vertex Sampling]
\label{def:modnested}
Fix an integer $\al$ where $1 \le \al \le \lgnone$. Using \Cref{lem:thorup} find the set $A_\al$. For each $\al+1 \le i \le \lgn-1$, $A_i$ is sampled from $A_{i-1}$ with probability $\prob{\two{i-1}}$. 
\end{definition}

To understand the advantage of this sampling process, consider the bounds it provides. First, \Cref{lem:ball} still holds for all $i \ge \al$ (see \Cref{proof:ball} for a complete proof). For $A_\al$, we can bound both $\mod{\ballal(s)}$ and $\mod{\clal(s)}$ for all $s \in V$. However, for any $i > \al$, we bypass \Cref{lem:thorup} and can therefore only bound $\mod{\ba_i(s)}$, not $\mod{\clu_i(s)}$. Our algorithm crucially leverages the fact that $\clal(\cdot)$ remains bounded. We now describe a simple property of our modified nested sampling.

\begin{lemma}
\label{lem:trivial}
Let $\al \le i \le \lgnone$. Let $z$ be a vertex with degree at least $\tl(\two{i})$. Then, with high probability, $z$ has a neighbor in $A_i$ where $A_i$ is constructed using modified nested vertex sampling.
\end{lemma}

Note that the above lemma is immediate when we use only nested vertex sampling. In  \Cref{proof:trivial}, we show that the lemma also holds under modified vertex sampling. We now move to the next important definition, which we adapt from \cite{Gup25}.

\begin{definition}
\label{def:closevertices}
Consider an $st$ path and let $i \ge \al$. Let $a_i$ be the vertex closest to $s$ on the $st$ path such that its distance to $pivot_i(a_i)$ is at most $1$. Define $u_i := pivot_i(a_i)$. Similarly, let $b_i$ be the vertex closest to $t$ on the $st$ path such that its distance to $pivot_i(b_i)$ is at most $1$. Define $v_i := pivot_i(b_i)$.
\end{definition}

\begin{figure}[hpt!]
\centering
\begin{tikzpicture}
\coordinate (s)  at (0,0);
\coordinate (ai) at (2,0);
\coordinate (bi) at (8,0);
\coordinate (aii) at (4,0);
\coordinate (bii) at (6,0);
\coordinate (t)  at (10,0);

\coordinate (ui) at (2.5,1);
\coordinate (vi) at (7.5,1);
\coordinate (uii) at (4,1);
\coordinate (vii) at (6,1);

\draw[thick] (s) -- (t);

\node[circle,fill,inner sep=1.5pt,label=below:$s$] at (s) {};
\node[circle,fill,inner sep=1.5pt,label=below:$a_i$] at (ai) {};
\node[circle,fill,inner sep=1.5pt,label=below:$b_i$] at (bi) {};
\node[circle,fill,inner sep=1.5pt,label=below:$a_{i+1}$] at (aii) {};
\node[circle,fill,inner sep=1.5pt,label=below:$b_{i+1}$] at (bii) {};
\node[circle,fill,inner sep=1.5pt,label=below:$t$] at (t) {};

\node[circle,fill,inner sep=1.5pt,label=above:$u_i$] at (ui) {};
\node[circle,fill,inner sep=1.5pt,label=above:$v_i$] at (vi) {};
\node[circle,fill,inner sep=1.5pt,label=above:$u_{i+1}$] at (uii) {};
\node[circle,fill,inner sep=1.5pt,label=above:$v_{i+1}$] at (vii) {};
\draw (ai) -- (ui);
\draw (bi) -- (vi);
\draw (aii) -- (uii);
\draw (bii) -- (vii);
\end{tikzpicture}
\caption{Figure shows the relative position of $a_i$'s, $u_i$'s, $b_i$'s and $v_i$'s along the $st$ path}
\end{figure}
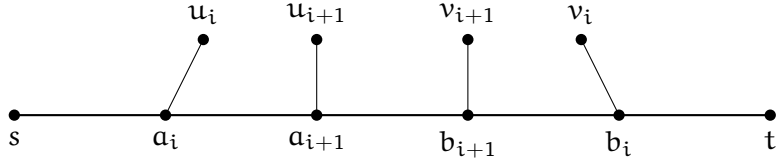

In \cite{Gup25}, the vertices $a_i,u_i,b_i,v_i$ are defined for all $i \ge 0$. Since our modified vertex sampling starts from $\al$, we define these terms only for $i \ge \al$. We first show that, for each $i \ge \al$, the vertex $a_i$ exists. For this, we rely on \Cref{ass:baswana}, which states that there is a vertex of degree at least $\tl(\sqrt n)$ on the $st$ path. Let $z$ be the vertex closest to $s$ with degree $\ge \tl(\sqrt n)$. By \Cref{lem:trivial}, with high probability, a vertex adjacent to $z$ lies in $A_{\lgnone}$. Thus, $pivot_{\lgn-1}(z)$ is at a distance $\le 1$ from $z$. Hence, we set $ a_{\lgnone} := z$ and $u_{\lgnone} := pivot_{\lgnone}(z)$.

Since vertex sampling is nested, $u_{\lgnone}$ which lies in $A_{\lgn-1}$ also lies in $A_{\lgn-2},\\ A_{\lgn-3}, \dots, A_{\al+1}, A_\al$.  This implies the existence of $a_i$ and $u_i$ for all $i \ge \al$ along the $st$ path. By a symmetric argument, the vertices $b_i$ and $v_i$ also exist for all $i \ge \al$ along the $st$ path. 

We now state another lemma that bounds the degree of the vertices on the $sa_i$ path and the $tb_i$ path for each $i \ge \al$. The lemma follows from the fact that $a_i$ is the vertex closest to $s$ whose distance to $pivot_i(a_i)$ is at most $1$. We formally prove this lemma in \Cref{proof:nearlowdegree}.

\begin{lemma}[\cite{Gup25}]
	\label{lem:nearlowdegree}
	For each $\al \le i\le \lgn -1$, the degree of each edge on the $sa_i$ and $tb_i$ paths is at most $\tl(\two{i})$ with a high probability.
\end{lemma}

Unlike \cite{Gup25}, we use $\fmm$ crucially in our algorithm. We now state the relevant results in this area.

\subsection{Fast Matrix Multiplication ($\fmm$)}

\begin{definition}[Distance product or min-plus product]
Let $A$ be a matrix of size $n \times n^r$ and $B$ be a matrix of size $n^r \times n$. The distance product or min-plus product of $A$ and $B$ is defined as
$$
(A \star B)_{ij} := \min_{1 \le \ell \le n^r} \{ A_{i\ell} + B_{\ell j} \}
$$
where $1 \le i,j \le n$.
\end{definition}

The distance product can be computed quickly using fast rectangular matrix multiplication. In our paper, we are satisfied even with an approximate distance product. A matrix $C$ is a $(1+\ep)$ approximation of $A \star B$ if, for each $i,j$, 
$$(A \star B)_{ij} \le C_{ij} \le (1+\ep)(A \star B)_{ij}$$
where $\ep > 0$. We use the following result implied  in \cite{Zwick02}, which computes a $(1+\ep)$ approximation of the distance product efficiently. Recently, \cite{DoryFKNWV24} used this result to get better algorithms for additive as well as multiplicative $\ap$.

\begin{lemma}[Implied by \cite{Zwick02}]
\label{lem:zwick}
Let $W$ be a positive integer and let $\ep > 0$ be a parameter. Consider matrices $A$ of size $n \times n^r$ and $B$ of size $n^r \times n$ where each entry of $A$ and $B$ lies in $\{0,1,\ldots,W\} \cup \{\infty\}$. Then there exists an algorithm that computes a $(1+\ep)$ approximation of the min-plus product $A \star B$ in time
$$
\tilde{O}\!\left(\frac{n^{\omega(r)}}{\ep}\log W\right).
$$
Here $\omega(r)$ is a function that depends on $r$.
\end{lemma}

In our setting, the maximum finite distance between any two vertices is bounded by $O(n)$. Thus, we can set $W = O(n)$, and the $\log W$ factor is  absorbed by the $\tl(\cdot)$ notation. Thus, we reduce approximate distance products to rectangular matrix multiplication. 
In our algorithm, the parameter $r$ is small. Specifically, in our dense case, the inner dimension of the matrix multiplication is $n^{0.125}$, which corresponds to $r=0.125$. Results in \cite{Coppersmith82,Williams14ACC} show that when $r \le 0.17227$ we have $\omega(r) = 2$.  Therefore, our algorithm runs in $\tl(n^2)$ time even though it uses $\fmm$.

%% file: assumption.tex
\section{Constant additive approximation  of $st$ path}
In the overview, we designed an algorithm that finds a constant additive approximation of the $st$ path. However, the algorithm relied crucially on \Cref{ass:overview:2}. We now present a  lemma that justifies this assumption.

\begin{lemma}
	\label{lem:mainclaim}
	Fix $\al = \lgn-4$ in \Cref{def:modnested}. There is a randomised algorithm $\bp$  that ensures that upon termination, for every pair of vertices $s,t \in V$, either $est(s,t) \le 2\mod{st}$, or both of the following inequalities hold:
	\begin{enumerate}[nosep]
		\item $\mod{s\,pivot_{\al}(s)} > \floor{\sttwo} - c'$, and
		\item $\mod{t\,pivot_{\al}(t)} > \floor{\sttwo} - c'$
	\end{enumerate}
	
	where $c'$ is a constant. Moreover, the running time of $\bp$ is $\tl(n^2)$.
\end{lemma}

We will prove this lemma in the next section, but let us first explain it intuitively. Our algorithm $\bp$ either already gets a 2-approximation of the $st$ path or it ensures that the distance between $s$ and $pivot_{\al}(s)$ (and symmetrically, between $t$ and $pivot_{\al}(t)$) is $\sttwo$, up to an additive constant $c'$. Intuitively, this means that the balls around $s$ and $t$ are very close to each other. Henceforth, we will assume that after termination of $\bp$, $est(s,t) > 2|st|$. Thus, the two inequalities in the lemma hold. We will use these two inequalities to justify \Cref{ass:overview:2}. To see this, let us redefine the notation $p_s$ and $p_t$ from the overview.

\begin{definition}
\label{def:xy}
Assuming $t \notin ball_\al(s)$ and $s \notin ball_\al(t)$, let $p_s$ denote the vertex on the boundary of $ball_{\alpha}(s)$ that lies on the $st$ path. Similarly, let $p_t$ denote the vertex on the boundary of $ball_{\alpha}(t)$ that lies on the $st$ path. Thus $|sp_s| = |s\,pivot_\al(s)| - 1$ and $|tp_t| = |t\,pivot_\al(t)| - 1$.
\end{definition}

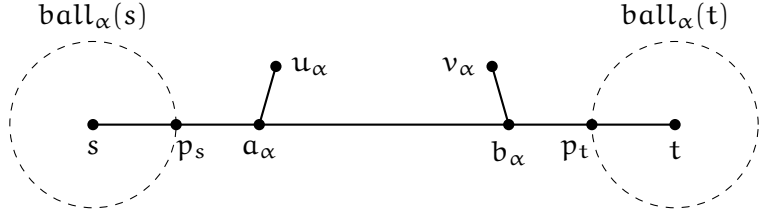
\begin{figure}[hpt!]
\centering
\begin{tikzpicture}[scale=1.1]
\draw[thick] (0,0) -- (7,0);
\draw[thick] (2,0) -- (2.2,0.7);
\draw[thick] (5,0) -- (4.8,0.7);

\node[circle,fill,inner sep=1.5pt,label=below:$s$] (s) at (0,0) {};
\node[circle,fill,inner sep=1.5pt,
      label={[xshift=6pt]below:$p_s$}] (ps) at (1,0) {}; \node[circle,fill,inner sep=1.5pt,label=below:$a_{\al}$] (a) at (2,0) {};
\node[circle,fill,inner sep=1.5pt,label=right:$u_{\al}$] (u) at (2.2,0.7) {};
\node[circle,fill,inner sep=1.5pt,label=below:$b_{\al}$] (b) at (5,0) {};
\node[circle,fill,inner sep=1.5pt,label=left:$v_{\al}$] (v) at (4.8,0.7) {};
\node[circle,fill,inner sep=1.5pt,
      label={[xshift=-6pt]below:$p_t$}] (pt) at (6,0) {}; \node[circle,fill,inner sep=1.5pt,label=below:$t$] (t) at (7,0) {};

\draw[dashed] (0,0) circle (1);
\node at (0,1.3) {$ball_\alpha(s)$};

\draw[dashed] (7,0) circle (1);
\node at (7,1.3) {$ball_\alpha(t)$};

\end{tikzpicture}
\caption{A diagram showing the relation between $p_s, a_\al,u_\al$ and $p_t,v_\al, b_\al$}
\end{figure}

If $t \in ball_\al(s)$, then using \Cref{lem:ball}, $est(s,t) = |st|$. Henceforth, we will assume that $t \notin ball_\al(s)$ and $s \notin ball_\al(t)$. We now show several properties of $p_s$ and $p_t$ that our algorithm will use. Recall that we defined the vertices $a_{\al}$ and $b_{\al}$ on the $st$ path. We now describe the relation between $p_s$ and $a_{\al}$ and between $p_t$ and $b_{\al}$.

\begin{lemma}
	\label{lem:ps}
Vertex $p_s$ lies closer to $s$ than $a_{\al}$, that is $|sp_s| \le |sa_{\al}|$. Similarly $|tp_t| \le |tb_{\al}|$. Moreover, the degree of each edge on the $sp_s$ and $tp_t$ paths is at most $\tl(\two{\al})$.
\end{lemma}
\begin{proof}
	Note that $u_{\alpha} \in A_{\alpha}$ and hence it is also a candidate for $pivot_{\alpha}(s)$. Hence $|s\,pivot_{\alpha}(s)| \le |su_{\alpha}|$. We now show that $p_s$ lies closer to $s$ than $a_{\alpha}$ on the $st$ path. Observe that:
	
	\begin{align*}
		|sp_s| &= |s~pivot_\al(s)|-1 \hspace{2cm}\text{(Since $p_s$ lies on the boundary of $ball_\al(s)$)}\\
			 &\le |s u_\al|-1  \hspace{3.25cm}\text{(Since $|s~pivot_\al(s)| \le |s u_\al|$)}\\
			 &\le (|sa_\al|+1)-1. \hspace{2.15cm}\text{(By triangle inequality)}\\
			 &\le|s a_\al| 
	\end{align*} 
	
	This proves the first part of the lemma and, by symmetry, also implies that $|tp_t| \le |tb_{\alpha}|$. We now prove the second part. By \Cref{lem:nearlowdegree}, every edge on the $sa_{\alpha}$ path has degree at most $\tl(\two{\alpha})$. Since $sp_s$ forms a subpath of $sa_{\alpha}$, every edge on $sp_s$ also satisfies this bound. The same argument applies to the $tp_t$ path.
\end{proof}

Next, we show that $p_s$ and $p_t$ are close to each other on the $st$ path justifying \Cref{ass:overview:2}.

\begin{lemma}
	\label{lem:xy}
Vertex $p_s$ lies closer to $s$ than $p_t$ on the $st$ path. Formally, $|sp_s| \le |sp_t|$. Moreover,  $|p_sp_t| \le c$ where $c = 2c'+1$.
\end{lemma}
\begin{proof}
By \Cref{lem:ps}, we have $|sp_s| \le |sa_{\alpha}|$. Hence $p_s$ lies on the $sa_{\alpha}$ path. Similarly, $p_t$ lies on the $tb_{\alpha}$ subpath of the $st$ path. If we show that $a_{\alpha}$ lies closer to $s$ than $b_{\alpha}$, we will prove the first part of the lemma. By definition, $a_{\alpha}$ is the closest vertex to $s$ on the $st$ path whose distance from $pivot_{\alpha}(a_{\alpha})$ is at most $1$. If $b_{\alpha}$ were closer, it would violate this definition. Hence $a_{\alpha}$ lies closer to $s$ than $b_{\alpha}$, or both vertices coincide. We now prove the second part of the lemma. Since $p_s$ lies closer to $s$ than $p_t$ on the $st$ path we obtain the following relation.

	\begin{align*}
		|p_sp_t| &= |st| - |sp_s| - |tp_t|\\
		\shortintertext{Since $p_s$ lies on the boundary of $ball_\al(s)$, $|sp_s| = |s~pivot_\al(s)|-1$. Similarly, $|tp_t| = |t~pivot_\al(t)|-1$. Thus, we get:}
			 & =|st| -(|s~pivot_\al(s)|-1) - (|t~pivot_\al(t)|-1) \\
			 \shortintertext{Using \Cref{lem:mainclaim}, $| s~pivot_\al(s)| \ge \floor{\sttwo} -c'+1 $. Similarly, $| t~pivot_\al(t)| \ge \floor{\sttwo} -c'+1 $. Substituting these, we get:}
			 &\le |st| - 2\left(\floor{\sttwo} - c'\right) \hspace{2cm}\\
			 &= |st| - 2\floor{\frac{|st|}{2}} + 2c'\\
			 &\le |st| - (|st|-1) + 2c'\\
			 & = 2c'+1\\
	\shortintertext{Setting $c = 2c'+1$, we get:}		
			& = c  
	\end{align*}	
\end{proof}

We now follow our approach as described in the overview. For each $x \in V$, we compute a set $\clo(x)$ as defined in the overview. For each candidate vertex $x$, we divide the algorithm into two cases based on the size of $\clo(x)$:

\subsection{Sparse Case: $|\clo(x)| \le n^\be$}
\label{sec:case1}
In this case, we repeat \Cref{alg:overview:trivial}. We first call the function $\bc$, in which we construct the set $\clo(x)$ for each $x \in V$. We then run \Cref{alg:modtrivial}, which is essentially \Cref{alg:overview:trivial} from the overview with two changes: we process $x$ only when $\clo(x)$ is small, and we pair $x$ only with vertices in $\clo(x)$.

\begin{minipage}{0.36\textwidth}
    \begin{algorithm}[H]
    \caption{\bc}
	\ForEach{$x \in V$}
	{	$\clo(x) \leftarrow \emptyset$\\
		\ForEach{$y \in V$}
		{
			\If{$est(x,y) \le 2c+1$}
			{
				add $y$ to $\clo(x)$
			}
		}
	}	
	\end{algorithm}
\end{minipage}
\begin{minipage}{0.60\textwidth}
    \begin{algorithm}[H]
    
	\ForEach{$x \in V$ such that $|\clo(x)| \le n^{\be}$}
	{	
		\ForEach{$y \in \clo(x)$}
		{
			\ForEach{$(s,t)$ in $\clal(x) \times \clal(y)$}
		{
			$
			est(s,t) = \min
			\begin{cases}
			est(s,t),\\
			est(s,x) + est(x,y) + est(y,t)
			\end{cases}
			$
		}
		}
	}	
	\caption{}
	\label{alg:modtrivial}
	\end{algorithm}
\end{minipage}

Using the analysis from the overview, and setting $x := p_s$ as in \Cref{alg:overview:trivial}, we obtain an additive $c+1$ approximation of the $st$ path using \Cref{eq:overview:approx}. We now analyze the running time of the algorithm. 

The procedure $\bc$ takes $O(n^2)$ time. Consider the running time of \Cref{alg:modtrivial}. The algorithm processes a vertex $x$ only when $|\clo(x)| \le n^{\be}$ and pairs $x$ only with vertices in $\clo(x)$. Hence, $x$ pairs with at most $n^{\be}$ vertices. Using \Cref{lem:thorup}, since $\clu_\al(\cdot) = \tl(\two{\al})$, the algorithm spends $\tl(\two{\alpha+1})$ time to process each pair. This is the only place in our algorithm where we need a bounded size of cluster. Thus, the running time of \Cref{alg:modtrivial} equals:
	\begin{align*}
		\sum_{x} \sum_{\substack{y\in \clo(x)}}& \tl(\two{\al+1})&\\
		&= \sum_x \tl(n^{0.875}\two{\al+1})\\
		&= \tl(n^{1.875}\two{\al+1})
		\shortintertext{In \Cref{lem:mainclaim}, $\al$ is set to $\lgn-4$. Thus, we get:}
		&= \tl(n^{1.875}\two{\lgn-3})\\
		&=\tl(n^{1.875}2^{\frac{\log n}{8}})\\
		&=\tl(n^{1.875} \times n^{0.125})\\
		& = \tl(n^2)
	\end{align*}

Thus, the total running time of our algorithm is $\tl(n^2)$. We now describe how the algorithm processes $x$ when $|\clo(x)|$ is large.

\subsection{Dense Case: $|\clo(x)| > n^\be$}

\label{sec:case2}

In this section, we use $\fmm$ to estimate the distance between $s$ and $t$. In \Cref{alg:fmm}, we compute a random set $L_{\ga}$ of vertices where each vertex of $G$ is sampled with probability $\tl(\frac{1}{\ga})$, where $\ga \in \{n^{\be}, 2n^{\be}, 2^{2}n^{\be}, \dots \}$. For each $z \in L_{\ga}$, we call $\uf(z,\ga)$. This function constructs a graph $G(\ga)$ that contains all edges whose degree is at most $\tl(\ga)$. One can verify that $G(\ga)$ contains at most $\tl(n\ga)$ edges. We then compute the shortest path from $z$ to every other vertex in $G(\ga)$ and update the estimates, if required.

\begin{algorithm}[H]
$\ga \leftarrow n^{\be}$\\
\While{$\ga \le n$}
{
$L_{\ga} \leftarrow$ a random sample of size $\tl(\frac{n}{\ga})$\\
\ForEach{$z \in L_{\ga}$}
{
$\uf(z,\ga)$
}
Construct a matrix $B$ of size $n \times \tl(\frac{n}{\ga})$ that stores the estimates of distances from vertices in $L_{\ga}$ to vertices in $G$\\
Using \Cref{lem:zwick}, compute $C = B \star B^T$\\
Update the estimates for all vertex pairs in $G$ using $C$, if required\\
$\ga \leftarrow 2\ga$
}
\caption{Compute estimates using $\fmm$}
\label{alg:fmm}
\end{algorithm}

\begin{algorithm}[H]
Make a graph $G(\gamma)$ that contains all edges of degree $\le \tl(\ga)$\\

Find the shortest path from $z$ to all other vertices in $G(\ga)$ and update $est(z,\cdot)$, if required
\caption{$\uf(z, \ga)$}
\label{alg:uf}
\end{algorithm}

After this step, we estimate the distance between every vertex pair using vertices in $L_{\ga}$. We construct a matrix $B$ of size $n \times \tl(\frac{n}{\ga})$. The rows correspond to vertices of $G$ and the columns correspond to vertices in $L_{\ga}$. For each entry $B_{ij}$, we store the estimate of the distance between vertex $i \in V$ and vertex $j \in L_{\ga}$ that we updated in \Cref{alg:uf}. We then compute $C = B \star B^T$ using \Cref{lem:zwick}. Note that $C$ provides a $(1+\ep)$ approximation of $B \star B^T$. We update the estimates for all vertex pairs using $C$, if required. This completes the description of \Cref{alg:fmm}. We now analyze the running time and the approximation guarantee of our algorithm.
  \subsubsection{Running Time}

 Fix a value of $\ga$. We first analyze the running time of $\uf(z,\ga)$. Since the size of $G(\ga)$ is $\tl(n\ga)$, the running time of $\uf(z,\ga)$ is $\tl(n\ga)$. The set $L_{\ga}$ contains $\tl(\frac{n}{\ga})$ vertices, so all invocations of $\uf$ together take $\tl(n^2)$ time.

	Next, we compute the distance product $B \star B^T$. The sizes of $B$ and $B^T$ are $n \times \tl(\frac{n}{\ga})$ and $\tl(\frac{n}{\ga}) \times n$ respectively. Since $\ga \ge n^{\be}$, we have $\frac{n}{\ga} \le n^{0.125}$. Using \Cref{lem:zwick}, we compute $C$ in time $\tl(\frac{n^{\omega(0.125)}}{\ep})$. We then update the estimates for all vertex pairs using $C$, if required. This step takes $O(n^2)$ time.
 
	Thus, the running time for a fixed $\ga$ equals $\tl(n^2 + \frac{n^{\omega(0.125)}}{\ep})$. The while loop repeats at most $\log n$ times, so the total running time remains: 	
	\begin{align*}
		&\tl\left(n^2 + \frac{n^{\omega(0.125)}}{\ep}\right)\\
		\shortintertext{Using \cite{Coppersmith82,Williams14ACC}, $\omega(0.125) =2$. Thus, we get:}
		=&~\tl\left(\frac{n^2}{\ep}\right)\\
		\shortintertext{ Setting $\ep = 0.5$, we get:}
		=&~ \tl(n^2)
	\end{align*}
	We now analyze the approximation guarantees of the algorithm.

\subsubsection{Approximation Ratio}

 We start with the following definition:
 \begin{definition}
 \label{def:trueclose}
 	Let $\tclo(p_s) = \{y ~|~ |p_sy| \le 2c+1\}$
 \end{definition}
 
  Informally, $\tclo(p_s)$ denotes the {\em actual} (not estimated) set of all vertices at distance at most $2c+1$ from $p_s$. We cannot compute $\tclo(p_s)$ as we do not know the exact distance between vertices in the graph. However,  we claim that a vertex in $\clo(p_s)$ must also lie in $\tclo(p_s)$.  \begin{lemma}
 \label{lem:containment}
 	If $y \in \clo(p_s)$, then $y \in \tclo(p_s)$. Thus, $|\tclo(p_s)| \ge |\clo(p_s)|$.
 \end{lemma}
 \begin{proof}
 	Since $\clo(p_s)$ is based on estimates that can only be larger than the actual distances, $est(p_s,y) \le 2c+1$ implies $|p_sy| \le 2c+1$. Thus, if $y \in \clo(p_s)$, then $y \in \tclo(p_s)$.
 \end{proof}
 \begin{figure}[hpt!]
\centering
\begin{tikzpicture}[scale=1.1]

\coordinate (ps) at (2,0);
\filldraw[fill=gray!20, draw=black, dashed] (ps) circle (2);

\node at (2,2.3) {$\tclo(p_s)$};

\draw (ps) -- ++(2,0);
\draw (ps) -- ++(1.4,1.4)
    node[midway, above, sloped] {\scriptsize distance = $2c+1$};

\node[fill=black, circle, inner sep=1.5pt, label=below:$s$] (s) at (-2,0) {};
\node[fill=black, circle, inner sep=1.5pt, label=below:$p_s$] at (ps) {};
\node[fill=black, circle, inner sep=1.5pt, label=below:$p_t$] (pt) at (3.5,0) {};
\node[fill=black, circle, inner sep=1.5pt, label=below:$t$] (t) at (8,0) {};

\draw[thick] (s) -- (t);

\node[fill=black, circle, inner sep=1.5pt, label=above:$z$] (z) at (0.75,0.5) {};

\draw[] (z) -- (ps);

\end{tikzpicture}
\caption{The figure illustrates the scenario discussed in this section. The vertex $p_t$ lies in $\tclo(p_s)$. Moreover, there exists a sampled vertex $z \in L_{\ga}$ that also lies in $\tclo(p_s)$.}
\label{fig:tc}
\end{figure}
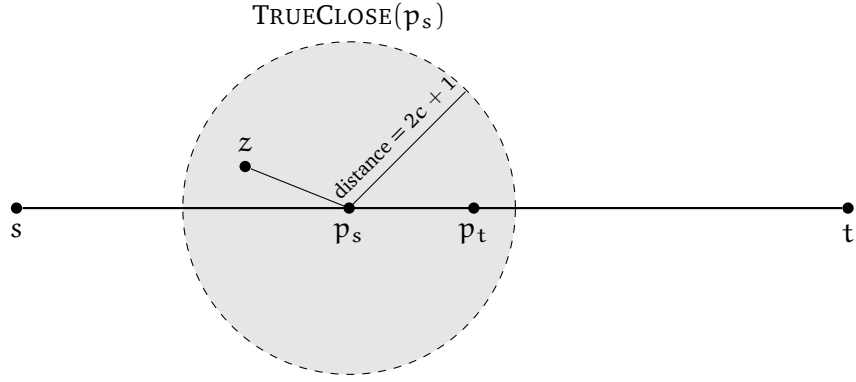
 Using the above lemma, since $|\clo(p_s)| > n^{\be}$, even $|\tclo(p_s)| > n^{\be}$. The set $\tclo(p_s)$ is fixed. If we sample vertices of suitable size at random, we are guaranteed to hit a vertex in $\tclo(p_s)$ with high probability.

 Assume that $\ga \le |\tclo(p_s)| \le 2\ga$ where $\ga \in \{n^{\be}, 2n^{\be}, 2^{2}n^{\be}, \dots\}$. In \Cref{alg:fmm}, we sample a set $L_{\ga}$ of size $\tl(\frac{n}{\ga})$. With high probability at least one vertex of $L_{\ga}$ lies in $\tclo(p_s)$. Without loss of generality, let this vertex be $z$. In $\uf(z)$, our algorithm computes the shortest path from $z$ to every other vertex in $G(\ga)$. We want $G(\ga)$ to contain both the shortest path $p_s\,z$ and the $st$ path (See \Cref{fig:tc}). This condition ensures a good approximation of the $st$ distance when we use the distance product. We now prove this property.

\begin{lemma}
\label{lem:containsx}
$G(\ga)$ contains the $p_s\,z$ path.
\end{lemma}

\begin{proof}
Recall that $|\tclo(p_s)| \le 2\ga$ and that $\tclo(p_s)$ contains all vertices within distance $2c+1$ from $p_s$. Hence, every vertex at distance strictly less than $2c+1$ from $p_s$ has degree at most $2\ga$, since its degree is bounded by $|\tclo(p_s)|$. Therefore, every vertex on the $p_s\,z$ path, except possibly $z$, has degree at most $2\ga$. In the construction of $G(\ga)$, we add all edges whose degree is at most $\tl(\ga)$. Thus, $G(\ga)$ contains the entire $p_s\,z$ path.
\end{proof}

\begin{lemma}
\label{lem:containst}
	$G(\ga)$ contains the $st$ path.
\end{lemma}

\begin{proof}
We show that $sp_s$, $p_sp_t$ and $p_t t$ lie in $G(\ga)$. This will imply that $st$ lies in $G(\ga)$. 
\begin{enumerate}
	\item $p_sp_t$ lies in $G(\ga)$
	
	By \Cref{lem:xy}, $|p_sp_t| \le c$. By \Cref{def:est}, $est(p_s,p_t) \le 2|p_sp_t|+1 \le 2c+1$, so $p_t \in \clo(p_s)$. Using \Cref{lem:containment}, we obtain $p_t \in \tclo(p_s)$. 

	Each vertex on the $p_sp_t$ path, except possibly $p_t$, has degree at most $2\ga$, since otherwise $|\tclo(p_s)| \ge 2\ga$. Thus, the degree of every edge on the $p_sp_t$ path is at most $\tl(\ga)$ and the entire path lies in $G(\ga)$.

	\item $sp_s$ and $p_tt$ lie in $G(\ga)$
	
	Using \Cref{lem:nearlowdegree}, every edge on the $sp_s$ and $tp_t$ paths has degree at most $\tl(\two{\al}) = \tl(\two{\lgn-4}) = \tl(n^{0.0625}) < \tl(n^{0.875}) \le \tl(\ga)$. Hence, the paths $sp_s$ and $tp_t$ also lie in $G(\ga)$. 
\end{enumerate}
Thus, the entire $st$ path lies in $G(\ga)$.
\end{proof}
  Since we compute the shortest path from $z$ to all other vertices in $G(\ga)$, we obtain a good estimate of $est(s,z)$ and $est(z,t)$ after completing $\uf(z,\ga)$. We prove the following.

\begin{lemma}
		\label{lem:szzt}
		$est(s,z) \le |sp_s| + 2c+1$ and $est(z,t) \le |p_st| + 2c+1$ after executing $\uf(z,\ga)$.

\end{lemma}
\begin{proof}
Using \Cref{lem:containsx} and \Cref{lem:containst}, $G(\ga)$ contain both $p_s\,z$ and $sp_s$ path (which is a subpath of $st$). Thus, after executing $\uf(z,\ga)$, we get a good estimate between $s$ and $z$.

\begin{align*}
    est(s,z) & \le |sp_s| + |p_s\,z| \notag\\
    \shortintertext{Since $z \in \tclo(p_s)$, it is at a distance at most $2c+1$ from $p_s$. Thus, we have}
             &\le |sp_s| + 2c+1
\end{align*}
We now bound $est(z,t)$. Again, using  \Cref{lem:containsx} and \Cref{lem:containst}, $G(\ga)$ contain both $p_s\,z$ and $p_s t$  path (which is a subpath of $st$). Thus, $est(z,t)$ can be bounded as:
\begin{align*}
    est(z,t) & \le |zp_s| + |p_s t| \\
    \shortintertext{As above, $|zp_s| \le 2c+1$. Thus, we have:}
             &\le 2c+1 +|p_st| 
\end{align*}
\end{proof}

In \Cref{alg:fmm}, we construct the matrix $B$, which contains the estimates between all vertices of $V$ and all vertices of $L_\ga$. Thus, it contains our revised estimates $est(s,z)$ and $est(z,t)$ calculated in the above lemma. Using \Cref{lem:zwick}, we get a $(1+\ep)$ approximation $C$ of $B\star B^T$. Thus, we get a new estimate between $s$ and $t$ using $C$ as follows:
\begin{align*}
    est(s,t) & \le (1+\ep) \min_{w \in L_\ga}\{est(s,w)+est(w,t)\}  \notag\\
             &\le (1+\ep)(est(s,z) + est(z,t))\\
             \shortintertext{Using \Cref{lem:szzt}, $est(s, z) \le |sp_s|+2c+1$ and $est(z,t) \le 2c+1+|p_st|$. Thus, we get:}
             &\le (1+\ep)((|sp_s| + 2c+1) + (2c+1+|p_st|))\\
             &= (1+\ep)(|st| + 4c+2) 
\end{align*}

If we set $\ep = 1/2$, then $est(s,t) \le 2|st|$ when $|st| \ge 12c+6$.

\subsection{Summary of the result in this section}
Our algorithm runs through two cases. In \Cref{sec:case1}, we described the first case when $|\clo(p_s)| \le n^\be$. In this case our algorithm ran with the running time of $\tl(n^2)$ and gave an additive $c+1$ approximation of $st$ path. This is a 2-approximation of $st$ path when $|st| \ge c+1$. In \Cref{sec:case2}, we described the second case when $|\clo(p_s)| > n^\be$. In this case our algorithm ran with the running time $\tl(n^2 )$ and gave a 2 approximation of $st$ path when $|st| \ge 12c+6$ (if we set $\ep = 1/2$). Overall our algorithm gives a 2-approximation of $st$ path when $|st| \ge 12c+6$ and runs in $\tl(n^2)$ time. 

Armed with this result, the proof of \Cref{thm:maintheorem} becomes extremely simple.

\section{Proof of \Cref{thm:maintheorem}}
We first run $\bp$ in $\tl(n^2)$ time. $\bp$ ensures that for each $s,t$ pair, either $est(s,t) \le 2|st|$ or the balls of $s$ and $t$ are close. If $est(s,t) \le 2|st|$, then we are already done. Otherwise, we use the guarantees provided by $\bp$ to compute a 2-approximation of the $st$ path if $|st| \ge 12c+6$. This is done in the previous section and its running time is $\tl(n^2)$. Thus, in total, the running time of our algorithm is $\tl(n^2)$, and it gives a 2-approximation of the $st$ path for all pairs that are at a distance of at least $12c+6$, where $c$ is a constant. This proves \Cref{thm:maintheorem}. The only remaining step is to prove the guarantees provided by $\bp$, which we do in the next section.

%% file: algo.tex
\section{Ensuring ball proximity}

In this section, we prove \Cref{lem:mainclaim}. To achieve this, we adapt two algorithms from \cite{Gup25} and run them as subroutines. Let us see the first subroutine from \cite{Gup25}.

\begin{lemma}[\cite{Gup25}]
\label{lem:es}
    There is an algorithm \textsc{EnsureCloseness} such that after its execution,
    for every pair $(s,t)$, either $est(s,t) \le 2|st|$, or,
    for every $0 \le i \le \lgn-1$, either
    \begin{enumerate}[label=(A\arabic*),nolistsep]
        \item\label{item:a1orig} $|sa_i| \le |tb_i|$ and $|sa_i|-|s~pivot_i(s)| \le 3$, or
        \item $|sa_i| \ge |tb_i|$ and $|tb_i|-|t~pivot_i(t)| \le 3$.
    \end{enumerate}
    Moreover, $\textsc{EnsureCloseness}$ runs in $\tl(n^2)$.
\end{lemma}

$\textsc{EnsureCloseness}$ plays a central role in \cite{Gup25}. After executing $\textsc{EnsureCloseness}$, either $est(s,t) \le 2|st|$ holds or we obtain a relation between $|sa_i|$ and $|s~pivot_i(s)|$ (or between $|tb_i|$ and $|t~pivot_i(t)|$). 
The algorithm \textsc{EnsureCloseness} performs $\lgn$ iterations, from $i = \lgn-1$ down to $i = 0$. After the $i$-th iteration completes, the relation in \Cref{lem:es} holds for that specific value of $i$. Once all iterations finish, the relation holds simultaneously for all $0 \le i \le \lgn-1$. In our paper, our modified nested vertex sampling starts at the intermediate set $A_\al$ rather than the base set $A_1$ as in \cite{Gup25}. Consequently, we execute a modified version of the algorithm, which we call \textsc{ModifiedEnsureCloseness}, running only from $i = \lgn-1$ down to $i = \al$ and skipping the iterations below $\al$.

\begin{lemma}
\label{lem:esmod}
    There is an algorithm \textsc{ModifiedEnsureCloseness} such that after its execution,
    for every pair $(s,t)$, either $est(s,t) \le 2|st|$, or,
    for every $\textcolor{red}{\bm{\al}} \le i \le \lgn-1$, either
    \begin{enumerate}[label=(A\arabic*),nolistsep]
        \item\label{item:a1mod} $|sa_i| \le |tb_i|$ and $|sa_i|-|s~pivot_i(s)| \le 3$, or
        \item\label{item:a2mod} $|sa_i| \ge |tb_i|$ and $|tb_i|-|t~pivot_i(t)| \le 3$.
    \end{enumerate}
    Moreover, $\textsc{ModifiedEnsureCloseness}$ runs in $\tl(n^2)$.
\end{lemma}

The only change from \Cref{lem:es} is the lower limit of the range: $\al \le i$ instead of $0 \le i$. The algorithm \textsc{ModifiedEnsureCloseness} and its proof remain identical to those in \cite{Gup25}. We now state the next important lemma from \cite{Gup25}, which crucially uses \Cref{lem:es}. Let us assume that $est(s,t) > 2|st|$ after executing $\textsc{EnsureCloseness}$. Then, \cite{Gup25} show the following lemma:

\begin{lemma}[\cite{Gup25}]
\label{lem:gupmain}
There is a randomised algorithm such that after its execution, for each $0 \le i \le \log \log n-1$,
$$est(u_i,v_i) \le |a_ib_i| + 18(\log \log n - i)$$

Moreover, the algorithm runs in $\tl(n^2)$ time with a high probability.
\end{lemma}

The algorithm runs for $\log \log n$ iterations, from $i=\log \log n-1$ to $0$. When $i=\log \log n-1$, it computes a good approximation between $u_{\lgn-1}$ and $v_{\lgn-1}$. In the $i$th iteration, it computes a good approximation between $u_i$ and $v_i$ using the relation between $u_{i+1}$ and $v_{i+1}$ by induction, and the relations among $sa_i$, $s\, pivot_i(s)$, and $tb_i$, $t\, pivot_i(t)$ in \ref{item:a1mod} and \ref{item:a2mod}. As in $\textsc{EnsureCloseness}$, we do not run the algorithm for $\lgn$ iterations. Instead, we run it for $\lgn-\al$ iterations and apply the modified \Cref{lem:esmod} to obtain the following.

\begin{lemma}
\label{lem:gupmod}
There is a randomised algorithm such that after its execution, for each $\textcolor{red}{\bm{\al}} \le i \le \log \log n-1$,
$$est(u_i,v_i) \le |a_ib_i| + 18(\log \log n - i)$$

Moreover, the algorithm runs in $\tl(n^2)$ time with a high probability.
\end{lemma}

Again, the algorithm and the proof of the above lemma remain the same as in \cite{Gup25}. The only change is that we run the algorithm of \cite{Gup25} for fewer iterations.  After we have run the above two algorithms,  we run our algorithm \bp, which is also present in \cite{Gup25} in a similar form.

The algorithm $\bp$ has three parts. The first two parts call $\uf(w)$ for every $w \in A_\al$. This two-pass approach is intentional: the first pass discovers and updates intermediate distance estimates, while the second pass leverages these newly updated estimates to find even shorter  paths.  The last part calls $\tri(s)$ for every $s \in V$, and estimates the distance from $s$ to $t$ using $pivot_\al(s)$.

\begin{minipage}{0.37\textwidth}
    \begin{algorithm}[H]
    {\color{blue}\tcc{Part 1}}
	\ForEach{$w \in A_\al$}
	{	
		$\uf(w)$
	}	
	{\color{blue}\tcc{Part 2}}
	\ForEach{$w \in A_\al$}
	{	
		$\uf(w)$
	}
	{\color{blue}\tcc{Part 3}}
	\ForEach{$s \in V$}
	{
		$\tri(s)$
	}
	\caption{\small{$\bp$}}
	\label{alg:estst}
	\end{algorithm}
\end{minipage}
\begin{minipage}{0.60\textwidth}
    \begin{algorithm}[H]
        \caption{$\uf(w)$}
		Construct a graph $G(w)$ that contains:\\
		(1) An edge from $w$ to every $x \in V$ with weight $est(w,x)$ \label{line:uf1}\\

		(2) All edges with degree $\tl(\two{\al})$\label{line:uf2}\\
		(3) For each vertex $x \in V$ and for each $i \in [\al \dots \lgn-1]$, an edge  $[x,pivot_i(x)]$ with weight $|x~pivot_i(x)|$\\
		Compute the shortest path from $w$ to all other vertices in $G(w)$ and update $est(w,\cdot)$ and $est(\cdot,w)$    \end{algorithm}
    \begin{algorithm}[H]
        \caption{$\tri(s)$}
		\ForEach{ $t \in V$}
		{
			$est(s,t) \leftarrow \min\{est(s,t), est(s,pivot_\al(s)) + est(pivot_\al(s),t) \}$\\
		}
    \end{algorithm}
\end{minipage}

We now describe the two subroutines used in $\bp$.
\begin{enumerate}
	\item $\uf(w)$
	
	In $\uf(w)$, we construct a graph $G(w)$ that contains:

\begin{enumerate}[nosep]
\item An edge from $w$ to every $x \in V$ with weight $est(w,x)$.
\item All edges with degree $\tl(\two{\al})$.
\item For each $x \in V$ and each $i \in [\al \dots \lgn-1]$, an edge from $x$ to $pivot_i(x)$ with weight $|x~pivot_i(x)|$.
\end{enumerate}

In (a) and (c) we add $\tl(n)$ edges to $G(w)$. In (b) we add $\tl(n\two{\al})$ edges to $G(w)$. 
We then compute shortest paths from $w$ in $G(w)$ and update the estimates from $w$ if needed. 
Thus, $\uf(w)$ runs in time $\tl(n\two{\al})$.

	\item $\tri(s)$
	
	The function $\tri(s)$ estimates the distance between $s$ and $t$ using $pivot_\al(s)$. It runs in $O(n)$ time.

\end{enumerate}

This completes the description of $\bp$. We now analyze its running time.
\subsection{Running Time}

In the first two parts of the algorithm, we call $\uf(w)$ for every $w \in A_\al$. The function $\uf(w)$ runs in time $\tl(n\two{\al})$. Hence the running time of the first two parts of  $\bp$ is $\tl(\sum_{w \in A_\al} n\two{\al}) = \tl(\frac{n}{\two{\al}} \cdot n\two{\al}) = \tl(n^2)$. Since $\tri(s)$ runs in time $O(n)$, Part 3 of $\bp$ takes $O(n^2)$. Thus $\bp$ runs in time $\tl(n^2)$. 

\subsection{Properties of $\bp$}
There are three parts of $\bp$. We now go over each part in detail.
\begin{enumerate}
	\item Part 1: Estimate between $v_\al$ and $pivot_\al(s)$
	
	In Part 1, we compute a better estimate between $v_\al$ and $pivot_\al(s)$. The algorithm processes every vertex in $A_\al$, including $v_\al$, and constructs $G(v_\al)$. We show that $G(v_\al)$ contains the following path from $v_\al$ to $pivot_\al(s)$:
	$$[v_\al,u_\al] \odot (u_\al,a_\al) \odot a_\al s \odot [s,pivot_\al(s)]$$
	\begin{center}
	\begin{tikzpicture}
    \node[draw, circle, fill=black, label=below:$s$, minimum size=4pt, inner sep=0pt] (s) at (-2,0) {};
    \node[draw, circle, fill=black, label=below:$t$, minimum size=4pt, inner sep=0pt] (t) at (9,0) {};

    \draw[thick] (s) -- (t);


    \node[draw, circle, fill=black, label=left:\scalebox{0.8}{$pivot_\al(s)$}, minimum size=4pt, inner sep=0pt] (pis) at (-1,1) {};
    \draw[thick] (s) -- (pis);

    \node[draw, circle, fill=black, label=above:\scalebox{0.8}{$u_\al$}, minimum size=4pt, inner sep=0pt] (ui) at (2, 0.5) {};
    \node[draw, circle, fill=black, label=below:\scalebox{0.8}{$a_\al$}, minimum size=4pt, inner sep=0pt] (ai) at (2, 0) {};
    \draw[thick] (ui) -- (ai);

    \node[draw, circle, fill=black, label=above:\scalebox{0.8}{$v_\al$}, minimum size=4pt, inner sep=0pt] (vi) at (7, 0.5) {};
	    \node[draw, circle, fill=black, label=below:\scalebox{0.8}{$b_\al$}, minimum size=4pt, inner sep=0pt] (bi) at (6.5,0) {};
    \draw[thick] (vi) -- (bi);

	    \draw[ thick] (ui) to[bend left=10] node[midway, above] {\scalebox{0.7}{$est(u_\al, v_\al)$}} (vi);
	    \draw[line width=1mm,pink,opacity=0.6] (ui) to[bend left=10]  (vi);
	    \draw[line width=1mm,pink,opacity=0.6] (ui)--(ai);
	    \draw[line width=1mm,pink,opacity=0.6] (ai)--(s);
	    \draw[line width=1mm,pink,opacity=0.6] (s)--(pis);

\end{tikzpicture}
\end{center}

	The first edge is added due to point (1) of $\uf(v_\al)$. The second and the last edges are added due to point (3). By \Cref{lem:nearlowdegree}, every edge in $sa_\al$ has degree $\le \tl(\two{\al})$, so point (2) of $\uf(v_\al)$ adds them. Since we compute shortest paths from $v_\al$ to all vertices, we update the estimate between $v_\al$ and $pivot_\al(s)$ as:

	\begin{align}
		est(v_\al, pivot_\al(s)) &\le est(u_\al,v_\al) + |(u_\al,a_\al)|+|a_\al s| +|s~pivot_\al(s)|\notag\\
		\shortintertext{ Using \Cref{lem:gupmod}, $est(u_\al,v_\al) \le |a_\al b_\al| + 18(\lgn - \al)$. Thus, we get:}
							 &\le \left(|a_\al b_\al| + 18(\lgn-\al)\right) + 1+|a_\al s| + |s~pivot_\al(s)|\notag\\
							 &=|sb_\al|+|s~pivot_\al(s)|+18(\lgn-\al)+1\label{eq:1}
	\end{align}
	
	By symmetry, we also get $est(u_\al, pivot_\al(t)) \le |ta_\al |+|t~pivot_\al(t)|+18(\lgn-\al)+1$.
	
	\item Part 2: Estimate between $pivot_\al(s)$ and $t$
	
	In Part 2, we compute a better estimate between $pivot_\al(s)$ and $t$. The algorithm processes every vertex in $A_\al$, including $pivot_\al(s)$, and constructs $G(pivot_\al(s))$. We show that $G(pivot_\al(s))$ contains the following path from $pivot_\al(s)$ to $t$:
	$$[pivot_\al(s),v_\al] \odot (v_\al,b_\al) \odot b_\al t$$
	
	\begin{center}
	\begin{tikzpicture}
    \node[draw, circle, fill=black, label=below:$s$, minimum size=4pt, inner sep=0pt] (s) at (-2,0) {};
    \node[draw, circle, fill=black, label=below:$t$, minimum size=4pt, inner sep=0pt] (t) at (9,0) {};

    \draw[thick] (s) -- (t);
    \node[draw, circle, fill=black, label=left:\scalebox{0.8}{$pivot_\al(s)$}, minimum size=4pt, inner sep=0pt] (pis) at (-1,1) {};
    \draw[thick] (s) -- (pis);

    \node[draw, circle, fill=black, label=above:\scalebox{0.8}{$v_\al$}, minimum size=4pt, inner sep=0pt] (vi) at (7, 0.5) {};
	    \node[draw, circle, fill=black, label=below:\scalebox{0.8}{$b_\al$}, minimum size=4pt, inner sep=0pt] (bi) at (6.5,0) {};
    \draw[thick] (vi) -- (bi);

	    \draw[ thick] (pis) to[bend left=10] node[midway, above] {\scalebox{0.7}{$est(pivot_\al(s), v_\al)$}} (vi);
	    \draw[line width=1mm,pink,opacity=0.6] (pis) to[bend left=10]  (vi);
	    \draw[line width=1mm,pink,opacity=0.6] (vi)--(bi);
	    \draw[line width=1mm,pink,opacity=0.6] (bi)--(t);
	\end{tikzpicture}
	\end{center}
	The first edge is added in point (1) of $\uf(pivot_\al(s))$. The second edge is added in point (3). By \Cref{lem:nearlowdegree}, every edge in $b_\al t$ has degree $\le \tl(\two{\al})$, so point (2) of $\uf(pivot_\al(s))$ adds them. After we compute shortest paths from $pivot_\al(s)$ in $G(pivot_\al(s))$, we update the estimate between $pivot_\al(s)$ and $t$ as:
	
	\begin{align}
		est(pivot_\al(s),t) &\le  est(pivot_\al(s),v_\al) + |(v_\al,b_\al)|+|b_\al t| \notag\\
		\shortintertext{ Using \Cref{eq:1}, we get:}
							 &\le (|sb_\al|+|s~pivot_\al(s)| +18(\lgn-\al)+1)+ 1+|b_\al t| \notag\\
							 &=|st|+|s~pivot_\al(s)|+18(\lgn-\al)+2 \label{eq:2}
	\end{align}
	
	By symmetry, we also get $est(pivot_\al(t),s) \le |st|+|t~pivot_\al(t)|+18(\lgn-\al)+2$.
	
	\item Part 3: Estimate the distance between $s$ and $t$
	
	In the third part of our algorithm, we estimate the distance between $s$ and $t$ via $pivot_\al(s)$ in $\tri(s)$. Thus, we get:
	\begin{align}
		est(s,t) &\le est(s,pivot_\al(s))+est(pivot_\al(s),t)  \notag\\
		\shortintertext{ Using \Cref{lem:ball}, we have calculated the distance between $s$ and $pivot_\al(s)$. Thus, $est(s,pivot_\al(s)) = |s~pivot_\al(s)|$. Also, using  \Cref{eq:2}, we get:}
							 &\le |s~pivot_\al(s)|+ (|st|+|s~pivot_\al(s)|+18(\lgn-\al)+2) \notag
 \\
							 &=|st|+2|s~pivot_\al(s)|+18(\lgn-\al)+2 \label{eq:3}
	\end{align}
	By symmetry, we also get $est(s,t) \le |st|+2|t~pivot_\al(t)|+18(\lgn-\al)+2$.
\end{enumerate}

\subsection{Proof of \Cref{lem:mainclaim}}
We now summarize the results of this section. We run the three algorithms in sequence: first \textsc{ModifiedEnsureCloseness} (\Cref{lem:esmod}), then 
algorithm in \Cref{lem:gupmod}, and finally $\bp$. Each runs in time $\tl(n^2)$, so the total running time is $\tl(n^2)$.

After \textsc{ModifiedEnsureCloseness} terminates, for every pair $(s,t)$ either $est(s,t) \le 2|st|$ holds, in which case we are done, or the conditions of \Cref{lem:esmod} hold. Henceforth, in this section we assume the latter, i.e., $est(s,t) > 2|st|$. The conditions of \Cref{lem:esmod} are then used as input by \Cref{lem:gupmod}, whose output is in turn used by  $\bp$ to obtain \Cref{eq:3}. We now show the guarantees given by $\bp$.  Substituting $\al=\lgn-4$ in \Cref{eq:3}, we obtain:

\begin{align*}
    est(s,t) \le |st| + 2|s\,pivot_{\lgn-4}(s)| + 74
\end{align*}

Suppose $|s\,pivot_{\lgn-4}(s)| \le \floor{|st|/2} - 37$. Then:
\begin{align*}
    est(s,t) &\le |st| + 2|s\,pivot_{\lgn-4}(s)| + 74 \\
             &\le |st| + 2\!\left(\floor{|st|/2} - 37\right) + 74 \\
             &\le 2|st|
\end{align*}

Thus, if $|s~pivot_{\lgn-4}(s)| \le \floor{|st|/2}- 37$, then we already obtain a $2$ approximation of the $st$ path. Using a symmetrical argument, if $|t~pivot_{\lgn-4}(t)| \le \floor{|st|/2} - 37$, we again obtain a $2$ approximation of the $st$ path. Thus, we have proven \Cref{lem:mainclaim} with $c' = 37$.

%% file: appendix.tex
\appendix
\section{Proof of \Cref{lem:thorup}}
\label{proof:thorup}
We construct the sample set $A_\al$ using an iterative approach based on the Thorup-Zwick algorithm \cite{ThorupZ01}. The construction relies on two primary procedures: $\fc$, which iteratively builds the sample set, and $\fbc$, which computes the corresponding balls and clusters.

In $\fc$, we initialize the active vertex set $W$ to $V$ and the sampled set $A_\al$ to $\emptyset$. We then proceed in rounds. In each iteration, we sample $\tl\left(\frac{n}{2^{2^\al}}\right)$ vertices from $W$ and add these to $A_\al$. Specifically, the subroutine $\smp(X,s)$ independently samples each element of $X$ with probability $\min\left\{1, \frac{s}{|X|}\right\}$.

Following the sampling step, we invoke $\fbc(A_\al,\al)$ to update $ball_\al(v)$ and $cluster_\al(v)$ for all $v \in V$. Any vertex $w$ whose cluster size exceeds the threshold of $4 \times 2^{2^\al}$ is retained in $W$ for the next iteration. This loop continues until no vertices have excessively large clusters and $W$ becomes empty. This completes the algorithm's description. We now state its guarantees.

\noindent
\begin{minipage}[t]{0.42\textwidth}
    \begin{algorithm}[H]
        \caption{$\fc$}
        \label{alg:fc}
        Initialize $A_\al \leftarrow \emptyset$\\ 
        $W \leftarrow V$\\
        \While{$W \ne \emptyset$}{
            $A_\al \leftarrow A_\al \cup \smp\left(W, \frac{n}{2^{2^\al}}\right)$\\
            \fbc($A_\al,\al$)\\
            $W \leftarrow \{w \in V \mid |cluster_\al(w)| > 4 \times 2^{2^\al}\}$
        }
        \Return $A_\al$
    \end{algorithm}
\end{minipage}
\hfill
\begin{minipage}[t]{0.54\textwidth}
    \begin{algorithm}[H]
        \caption{$\fbc(A_i,i)$}
        \label{alg:fbc}
        Perform a multi-source BFS from vertices of $A_i$ and update $pivot_i(s)$ and $est(s,pivot_i(s))$ for each $s \in V$\\
        \ForEach{$s \in V$}{
            Perform a BFS from $s$ up to distance  $|s~pivot_i(s)| - 1$ from $s$\\
            \ForEach{$v \in $ BFS Tree of $s$}{
            	$est(s,v) \leftarrow $ level of $v$ in BFS tree of $s$\\
                $ball_i(s) \leftarrow ball_i(s) \cup \{v\}$\\
                $cluster_i(v) \leftarrow cluster_i(v) \cup \{s\}$
            }
        }
    \end{algorithm}
\end{minipage}

Observe that in the first iteration of the while loop in $\fc$, the set $W = V$, so $\smp$ independently samples each vertex of $V$ with probability $\frac{1}{\two{\al}}$. We now demonstrate that after just this first iteration, $|\ballal(v)| \le \tl(\two{\al})$ for all $v \in V$.

To see why, order the vertices of the graph by their distance from $v$, breaking ties arbitrarily. Let $S$ denote the set consisting of the first $\tl(\two{\al})$ vertices in this sorted sequence. The probability that $A_\al$ contains no vertices from $S$ after the first iteration is $\left(1-\frac{1}{\two{\al}}\right)^{\tl(\two{\al})}$. By appropriately choosing the hidden constant in the $\tl$ notation, this probability is at most $\frac{1}{n^c}$ for some constant $c > 0$.

Consequently, with high probability, at least one vertex in $S$ is sampled into $A_\al$. Because the ball only explores up to the nearest sampled vertex (the pivot), the ball of $v$ will contain at most the vertices in $S$, meaning its size is bounded by $\tl(\two{\al})$. In subsequent iterations of $\fc$, the sampled set $A_\al$ can only grow. As $A_\al$ expands, the size of $ball_\al(v)$ can only decrease. Thus, we establish the following lemma:

\begin{lemma}
    \label{lem:ball_size}
    After running $\fc$, for each $v \in V$, $|\ballal(v)| \le \tl(\two{\al})$ with high probability. This establishes condition~(2) of \Cref{lem:thorup} for the ball size.
\end{lemma}

We now use the above lemma to bound the running time of $\fbc$.
\begin{lemma}
\label{lem:fbctime}
The running time of $\fbc(A_\al,\al)$ is $\tl(n^2 )$.	
\end{lemma}
\begin{proof}
In the first step of $\fbc(A_\al,\al)$, we perform a multi-source BFS from all vertices of $A_\al$. To this end, we add a new source vertex $\mathcal{X}$ to the graph and connect it to all vertices of $A_\al$ with an edge of weight 0. Then, we do a BFS from $\mathcal{X}$. 
Using this process, we can find $pivot_\al(s)$ and $|s\,pivot_\al(s)|$ for every $s \in V$. The time taken by this multi-source BFS is  $\tl(m) = \tl(n^2)$. 

Subsequently, for each $s \in V$, a BFS is performed up to distance $|s\,pivot_\al(s)| - 1$ from $s$, covering precisely $ball_\al(s)$ and recording exact distances, so $est(s,v) = |sv|$ for each $v \in ball_\al(s)$. Clusters are maintained alongside the balls.

 We now bound the number of edges explored in the BFS for a single vertex $s$. To compute this BFS, we only need to process vertices  strictly inside the $ball_\al(s)$ or at a distance $\le |s~pivot_\al(s)|-2$ from $s$. Every vertex strictly inside $ball_\al(s)$ has degree at most $\tl(\two{\al})$ as otherwise $|ball_\al(s)| \ge \tl(\two{\al})$ contradicting \Cref{lem:ball_size}. Thus, the BFS from $s$ processes at most $\tl(\two{\al})$ vertices each having degree $\tl(\two{\al})$. Thus, the BFS of $s$ takes  $\tl(\two{\al+1})$ time. Summing over all vertices, the overall time taken by $\fbc(A_\al,\al)$ is  $\tl(n \cdot 2^{2^{\al+1}}) = O(n^2)$ time (since $\al \le \lgn-1$).
\end{proof}

We now state a lemma from Thorup and Zwick \cite{ThorupZ01} that bounds the maximum cluster size and the number of iterations required by $\fc$. We refer the reader to \cite{ThorupZ01} for the complete proof.

\begin{lemma}[Proof of Theorem 3.1, \cite{ThorupZ01}]
\label{lem:tzbounds}
For every vertex $w \in V$, $|cluster_{A_\al}(w)| = O(\two{\al})$ deterministically. Furthermore, $\fc$ runs for an expected $\tl(1)$ iterations.
\end{lemma}

We now discuss the guarantees given by $\fc$
\begin{enumerate}[nosep, label=(B\arabic*)]
	\item\label{item:1} Size of $A_\al$
	
	In each iteration of $\fc$, we sample a set of size $\tl\left(\frac{n}{\two{\al}}\right)$ from $W$. Using \Cref{lem:tzbounds}, since $\fc$ runs an expected $\tl(1)$ times, the total number of vertices added to $A_\al$ is $\tl\left(\frac{n}{\two{\al}}\right)$ in expectation. Therefore, the preceding lemma bounds the size of $A_\al$ only in expectation. For our purposes, however, we need to bound the size of $A_\al$ with high probability.
	
	\item\label{item:2} Size of $ball_\al(\cdot)$ and $cluster_\al(\cdot)$
	
	Using \Cref{lem:ball_size}, ball size is bounded with high probability.  And by \Cref{lem:tzbounds}, cluster size is bounded deterministically.
	
	\item\label{item:3} Running time of $\fc$
	
	In a single iteration of the while loop in $\fc$, we perform three main operations:
	\begin{enumerate}[nosep]
	\item Sample vertices from $W$. This takes $O(n)$ time.
	\item Execute $\fbc(A_\al,\al)$. Using \Cref{lem:fbctime}, this step takes $\tl(n^2)$ time.
	\item Recompute the  set $W$. Since the clusters have already been computed, this step requires only $O(n)$ time.
	\end{enumerate}
	
	Overall, a single iteration takes $\tl(n^2)$ time. By \Cref{lem:tzbounds}, the loop in $\fc$ runs for an expected $\tl(1)$ iterations. Therefore, the total running time of $\fc$ is $\tl(n^2)$ in expectation. 
\end{enumerate}

At this stage, we nearly satisfy the conditions of \Cref{lem:thorup}. However, \ref{item:1}  and \ref{item:3} currently hold only in expectation, whereas we require both the size of $A_\al$ and the overall running time to be bounded with high probability. We now present a modified procedure, $\ofc$, designed to achieve these high-probability guarantees.

\begin{algorithm}[H]
    \caption{$\ofc$}
    \label{alg:ofc}
        \While{true} %
    {
        Run $\fc$, terminating early if it runs more than $\tl(1)$ iterations\\ %
        \If{$\fc$ completes within $\tl(1)$ iterations}
        {
            $A_\al \leftarrow$ the set returned by the above invocation of $\fc$\\
            \Return $A_\al$
        }
    }
    
\end{algorithm}

The algorithm $\ofc$ is straightforward. It invokes $\fc$ but restricts it to a maximum of $\tl(1)$ iterations, forcefully terminating the procedure if this limit is exceeded. If $\fc$ completes naturally within this threshold, $\ofc$ returns the set $A_\al$ produced by that run of $\fc$. Otherwise, it repeats the process. We now formalize and prove the guarantees provided by this modified algorithm.

Let $X$ denote the number of iterations executed by the while loop in $\fc$. By \Cref{lem:tzbounds}, $\ex[X] = \tl(1)$. Applying Markov's inequality, $\Pr[X \ge 2\ex[X]] \le 1/2$. We call an execution of $\fc$ \emph{successful} if it terminates within $2\ex[X] = \tl(1)$ iterations. Any single invocation of $\fc$ is therefore successful with probability at least $1/2$.

 $\ofc$ repeatedly invokes $\fc$ independently. The probability that $\tl(1)$ consecutive runs all fail is at most $(1/2)^{\tl(1)} \le n^{-c}$ for some constant $c > 0$. Hence, with high probability, there is a successful run of $\fc$ within the first $\tl(1)$  iterations of $\ofc$. We now show the guarantees given by $\ofc$.
 
 \begin{enumerate}
 	\item The size of $A_\al$
 	
 	We have shown that within the first $\tl(1)$ iterations of the while loop in $\ofc$, there is a successful run of $\fc$. Consider this successful run of $\fc$. It executes $\tl(1)$ iterations of $\fc$, and in each iteration, at most $\tl\lb\frac{n}{\two{\al}}\rb$ vertices are added to $A_\al$. Thus, the size of $A_\al$ at termination is at most $\tl(1) \times \tl\lb\frac{n}{\two{\al}}\rb = \tl\!\left(\frac{n}{\two{\al}}\right)$.
 	
 	\item Size of $ball_\al(\cdot)$ and $cluster_\al(\cdot)$
	
	Using \Cref{lem:ball_size}, ball size is bounded with high probability.  And by \Cref{lem:tzbounds}, cluster size is bounded deterministically.
	
	\item Running time of $\ofc$
	
	In $\ofc$, we run $\fc$ at most $\tl(1)$ times with a high probability. Also, we allow each invocation of $\fc$ to run at most $\tl(1)$ iterations. Using \ref{item:3}, each iteration of $\fc$ takes $\tl(n^2)$ time.  Thus, the running time of $\ofc$ is $\tl(n^2)$ with a high probability.
 \end{enumerate}
 Thus, all conditions of \Cref{lem:thorup}  hold with high probability, completing the proof.

\section{Proof of \Cref{lem:trivial}}
\label{proof:trivial}

To prove \Cref{lem:trivial}, we require that for each $i \ge \al$, the algorithm samples vertices into $A_i$ independently. However, since $\fc$ performs sampling across multiple iterations, independence within $A_\al$, and hence $A_i$, does not follow immediately.

To address this issue, consider the vertices selected in the first iteration of $\fc$ that lie in $A_\al$. Denote this set by $A_\al(1)$. $\fc$ samples each vertex from $V$ into $A_\al(1)$ independently with probability $\prob{\two{\al}}$. For $i > \al$, define $A_i(1) \subseteq A_i$ as the set of vertices obtained by sampling from $A_{i-1}(1)$. 

Let $R$ be the event that an invocation of $\fc$ finishes within the cutoff. Using our arguments in the previous section, $\Pr[R]\ge 1/2$. Since $\ofc$ returns the first independent invocation satisfying $R$, its output is distributed as an invocation of $\fc$ conditioned on $R$. Before conditioning, every vertex belongs to $A_i(1)$ independently with probability $1/\two{i}$ for every $i\ge\al$. This follows by induction: it holds for $i=\al$ by the first iteration of $\fc$, and
\[
 \Pr[v\in A_{i+1}(1)]
 =
 \Pr[v\in A_i(1)]\frac1{\two{i}}
 =
 \frac1{\two{i}}\frac1{\two{i}}
 =
 \frac1{\two{i+1}}.
\]

Consequently, for every fixed $S\subseteq V$ and every $i\ge\al$,
\begin{align*}
\Pr[A_i(1)\cap S=\emptyset\mid R]
&\le
\frac{\Pr[A_i(1)\cap S=\emptyset]}{\Pr[R]}\\
&\le
2\left(1-\frac1{\two{i}}\right)^{|S|}\\
&\le
2\exp\left(-\frac{|S|}{\two{i}}\right).
\end{align*}
Thus, if $|S|=\tl(\two{i})$, then $A_i(1)$, and hence
$A_i$, intersects $S$ with high probability. 
\section{Proof of \Cref{lem:ball}}
\label{proof:ball}
For our context $i \ge \al$. Let us now prove the bounds in the lemma.

\begin{enumerate}
	\item Ball size bound
	 
	The proof follows the same structure as that of \Cref{lem:trivial}; we highlight only the necessary modifications here. Let $S$ be the set of the $\tl(2^{2^i})$ vertices closest to $s$ in $G$. Applying the same probabilistic argument used in \Cref{lem:trivial}, we can show that with high probability, at least one vertex from $S$ is sampled into $A_i$. Consequently, the size of the ball is bounded, that is, $|ball_i(s)| \le |S| = \tl(2^{2^i})$.
	
	 \item  Running time and the exact distance property for vertices in the ball and to pivots
	 
	 We run $\fbc(A_i,i)$ for all $\al \le i \le \lgn-1$. Using an argument similar to \Cref{lem:fbctime}, the running time of $\fbc(A_i,i)$ is  $\tl(n^2)$ for a fixed $i \le \lgn-1$. Over all $i$, the running time remains $\tl(n^2)$. 
	 
	 We first run a multisource BFS from all vertices in $A_i$ in $\fbc(A_i,i)$. This computes, for every vertex, its pivot and its distance to that pivot. Next, for each $s \in V$, we run a BFS up to distance $|s~pivot_i(s)|-1$. This step computes the exact distance from $s$ to every vertex in $ball_i(s)$. Consequently, after executing $\fbc(A_i,i)$, we obtain $est(s,v)=|sv|$ for every $v \in ball_i(s)$.
\end{enumerate}

This completes the proof of \Cref{lem:ball}.

\section{Proof of \Cref{lem:nearlowdegree}}
\label{proof:nearlowdegree}	
The lemma follows directly from \Cref{lem:trivial}. We prove the claim for the $sa_i$ path; the argument for the $tb_i$ path is symmetric.

Assume that no vertex on the $sa_i$  path (except may be $a_i$) lies in $A_i$, since otherwise it would serve as a candidate for $a_i$ or $u_i$. Suppose, for contradiction, that there exists a vertex $z \neq a_i$ on the $sa_i$ path with degree $> \tl(\two{i})$. By \Cref{lem:trivial}, $z$ has a neighbor in $A_i$ with high probability. Since $z$ lies on the $sa_i$ path, we have $\mod{sz} < \mod{sa_i}$. Hence, $z$ is closer to $s$ than $a_i$ and satisfies $\mod{z\,pivot_i(z)} \le 1$. This contradicts \Cref{def:closevertices}, which defines $a_i$ as the closest such vertex.

Therefore, every vertex on the $sa_i$ path, except possibly $a_i$, has degree at most $\tl(\two{i})$. Consequently, every edge on the $sa_i$ path  has degree $\tl(\two{i})$. This proves \Cref{lem:nearlowdegree}.